\documentclass{article}
\makeatletter
\let\cl@chapter\undefined
\makeatletter

\usepackage{graphicx}
\usepackage{amsmath,amsthm, amssymb,enumerate}
\usepackage{cite}

\usepackage{tikz}
\usepackage{cleveref}
\usepackage{algorithm,pgfplots}
\usetikzlibrary{shapes,arrows,fit,calc,positioning,patterns,decorations.pathmorphing,decorations.pathreplacing}

\tikzset{ brokenrect/.style={
    append after command={
      \pgfextra{
      \path[draw,#1]
       decorate[decoration={zigzag,segment length=0.3em, amplitude=.7mm}]
       {(\tikzlastnode.north east)--(\tikzlastnode.south east)}      
        -- (\tikzlastnode.south west)|-cycle;
        }}}}
\tikzset{ brokenrect2/.style={
    append after command={
      \pgfextra{
      \path[draw,#1]
       decorate[decoration={zigzag,segment length=0.3em, amplitude=.7mm}]
       {(\tikzlastnode.north west)--(\tikzlastnode.south west)}      
        -- (\tikzlastnode.south east)|-cycle;
        }}}}

\newenvironment{varalgorithm}[1]
  {\algorithm\renewcommand{\thealgorithm}{#1}}
  {\endalgorithm}

\newtheorem{theorem}{Theorem}
\newtheorem{lemma}{Lemma}
\newtheorem{claim}{Claim}

\newtheorem{remark}{Remark}

\newtheorem{proposition}{Proposition}
\newtheorem{observation}{Observation}
\newtheorem{example}{Example}

\newcommand{\J}{\mathcal{J}}
\newcommand{\T}{\mathcal{T}}
\newcommand{\R}{\mathcal{R}}
\newcommand{\Q}{\mathcal{Q}}
\renewcommand{\O}{\mathcal{O}}
\newcommand{\States}{L}

\newcommand{\Fmax}{F_{\max}}

\newcommand{\psum}{p_{\mathrm{sum}}}

\newcommand{\new}[1]{{\color{black}#1}}

\newcommand{\jsi}{\new{B_{(i)}}}
\newcommand{\js}{\new{B}}
\begin{document}

\title{Joint replenishment meets scheduling
}


\author{P\'eter Gy\"orgyi$^1$ \and Tam\'as~Kis$^{1,*}$ \and  T\'{\i}mea~Tam\'asi$^{1,2}$ \and J\'ozsef B\'ek\'esi$^3$ 
}



\maketitle

\begin{abstract}
In this paper we consider a combination of the joint replenishment problem (JRP) and single machine scheduling with release dates. There is a single machine and one or more item types. Each  job  has a release date, a positive processing time, and it requires a subset of items.
A job can be started at time $t$ only if all the required item types were replenished between the release date of the job and time point $t$. The ordering of item types for distinct jobs can be combined.
The objective is to minimize the total ordering cost plus a scheduling criterion, such as total weighted completion time or maximum flow time, where  the cost of ordering a subset of items simultaneously is the sum of a joint ordering cost, and an additional item ordering cost for each item type in the subset.
We provide several complexity results for the offline problem, and  competitive analysis for online variants with min-sum and min-max criteria, respectively.

\end{abstract}

\let\thefootnote\relax\footnotetext{
This work has been supported by the National Research, Development and Innovation Office -- NKFIH, grant no.~SNN~129178, and ED\_18-2-2018-0006. The research of P\'eter Gy\"orgyi was supported by the J\'anos Bolyai Research Scholarship of the Hungarian Academy of Sciences.

$^1$  Institute for Computer Science and Control, E\"{o}tv\"{o}s Lor\'{a}nd Research Network, Kende Str. 13-17., Budapest 1111, Hungary\\
              Tel.: +3612796156\\
              Fax: +3614667503\\
$^2$  Department of Operations Research, Institute of Mathematics, ELTE E\"otv\"os Lor\'and University, Budapest, Hungary \\
$^3$ Department of Computer Algorithms and Artificial Intelligence, Faculty of Science and Informatics, University  of Szeged, Hungary \\
              Email: gyorgyi.peter@sztaki.hu,\ kis.tamas@sztaki.hu,\ timea.tamasi@sztaki.hu,\ bekesi@inf.u-szeged.hu         
\\
$^*$ corresponding author
}

\section{Introduction}\label{sec:intro}
The Joint Replenishment Problem (JRP) is a classical problem of supply chain management with several practical applications.
In this problem a number of  demands emerge over the time horizon, where each demand has an arrival time, and an item type (commodity). 
To fulfill a demand, the required item must be ordered {\em not sooner\/} than the arrival time of the demand.
Orders of different demands can be combined, and the cost of simultaneously ordering a subset of item types incurs a joint ordering cost and an additional item ordering cost for each item type in the order. None of these costs depends on the number of units ordered. 
Thus the total ordering cost can be reduced by combining the ordering of item types of distinct demands. However, delaying a replenishment of an item type delays the fulfillment of all the demands that require it. The objective function expresses a trade-off between the total ordering cost, and the cost incurred by delaying the fulfillment of some demands, see e.g., \cite{buchbinder13}. There are other variants of this basic problem, see Section~\ref{sec:lit_rev}.

In the above variant of JRP a demand becomes ready as soon as the required item type has been replenished after its arrival time.
However, in a make-to-order manufacturing environment the demands may need some processing by one or more machines before they become ready, or in other words, a scheduling problem emerges among the demands for which the required item types have been ordered.
This paper initiates the systematic study of these variants. 
By adopting the common terminology of scheduling theory, from now on we call the demands and item types  'jobs' and 'resource types', respectively.

We consider the simplest scheduling environment, the single machine. 
It means that there is one machine, which can process at most one job at a time.
Each job has a release date and  a processing time, and it requires one or more resource types (this is a slight generalization of JRP, where each demand has only one item type). 
A job becomes \new{ready to be started} only if all the required resources are replenished after its release date. We only consider non-preemptive scheduling problems in this paper.
The objective function is the sum of the replenishment costs and the scheduling cost (such as the total weighted completion time, the total weighted flow time, or the maximum flow time  of the jobs). 
Observe that if we delay a replenishment of a resource type $i$, then all the jobs waiting for $i$ will occupy the machine after $i$ arrives, thus the processing of these jobs may delay the processing of some other jobs arriving later.
In other words, at each time moment we have to take care of the jobs with release dates in the future, which is not easy in the offline case, and almost impossible in an online problem, where the jobs are not known in advance.

After a literature review in Section~\ref{sec:lit_rev}, we provide a formal  problem statement along with a brief overview of the main results of the paper in Section~\ref{sec:prob_statement}.
Some hardness results for min-sum, and min-max type criteria are presented in Section~\ref{sec:hardness_sum}, and \ref{sec:hardness_fmax}, respectively.
Polynomial time algorithms are described for some variants of the problem in Sections~\ref{sec:dyn_prog} and \ref{sec:poly_fmax}.
Then, in Sections~\ref{sec:online_min-sum} and \ref{sec:online_fmax} we provide online algorithms for min-sum type criteria, and for the $F_{\max}$ objective function, respectively. We conclude the paper in Section~\ref{sec:conclude}.

\section{Literature review}\label{sec:lit_rev}

The first results on JRP are more than 50 years old, see e.g., Starr and Miller~\cite{starr1962}, for an overview of older results we refer the reader to Khouja and Goyal~\cite{khouja2008}. 
Since then, several theoretical and practical results appeared, this review cites only the most relevant literature.

Over the years, a number of different variants of the JRP have been proposed and studied, some of which are mathematically equivalent. Originally, JRP was an inventory management problem, where demands have due dates, and to fulfill a demand, the required item type must be \new{ordered}  {\em before\/} the due-date of the demand. However, keeping the units on stock incurs an inventory holding cost beyond the total ordering costs, and we call this variant JRP-INV.  
In another variant, demands have release dates and deadlines, and they can be fulfilled by ordering the corresponding items in the specified time intervals. The cost of a solution is determined solely by the total ordering costs. We will refer to this variant as JRP-D.
In the last variant considered in this review, demands have a delay cost function instead of deadlines, which determines the cost incurred by the delay between the release date of the demand and the time point when it is fulfilled by an order. 
The objective function balances the total cost incurred by delaying the fulfillment of the demands and the total ordering costs. We  call this variant JRP-W. In this paper we combine JRP-W with a scheduling problem.

The complexity of JRP-INV is studied by Arkin et al.~\cite{arkin1989}, who proved strong NP-hardness.
Levi et al.~\cite{levi2006} give a 2-approximation algorithm, and the approximation ratio is reduced to 1.8 by Levi and Sviridenko~\cite{levi2006improved}, see also Levi et al.~\cite{levi2008}.
Cheung et al.~\cite{cheung2016} describe approximation algorithms for  several variants under the assumption that the ordering cost function is a monotonically increasing, submodular function over the item types, see also Bosman and Olver~\cite{bosman2020}.

The JRP-D is shown NP-hard in the strong sense by Becchetti, et al.~\cite{becchetti2009latency} as claimed by Bienkowski et al.~\cite{bienkowski2015}, and APX-hardness is proved by Nonner and Souza~\cite{nonner2009}, \new{who} also describe an $5/3$-approximation algorithm.
Bienkowski et al.~\cite{bienkowski2015} provide an 1.574 approximation algorithm, and new lower bounds for the best possible approximation ratio. For the special case when the demand periods are of equal length, they give an 1.5-approximation algorithm, and prove a lower bound of 1.2.
In Bienkowski et al.~\cite{bienkowski2014}, the online version of JRP-D is studied, and an optimal 2-competitive algorithm is described.

The NP-hardness of JRP-W with linear delay cost functions follows from that of JRP-INV (reverse the time line). This has been sharpened by Nonner and Souza~\cite{nonner2009} by showing that the problem is still NP-hard even if each item admits only three distinct demands over the time horizon.
Buchbinder et al.~\cite{buchbinder13} study the online variant of the problem providing a 3-competitive algorithm along with a lower bound of 2.64 for the best possible competitive ratio.
Bienkowski et al.~\cite{bienkowski2014} provide a 1.791-approximation algorithm for the offline problem, and they also prove a lower bound of 2.754 for the best possible competitive ratio of an algorithm for the online variant of the problem with linear delay cost function.

For scheduling problems we use the well-known $\alpha|\beta|\gamma$ notation of Graham et al.~\cite{graham1979}.
$1|r_j|\sum C_j$ is a known strongly NP-hard problem, see, e.g., Lenstra et al.~\cite{lenstra1977}. 
There is a polynomial time approximation scheme (PTAS) for this problem even in the case of general job weights and parallel machines (Afrati et al.~\cite{Afrati1999}).
However, there are other important approximation results for this problem.
Chekuri et al.~\cite{chekuri2001} provide an $e/(e-1)\approx 1.5819$-approximation algorithm  based on the $\alpha$-point method.
For the weighted version of the same problem Goemans et al.~\cite{goemans2002} present an $1.7451$-approximation algorithm.
If each job has the same processing time then the weighted version of the problem is solvable in polynomial time even in case of constant number of parallel machines (Baptiste \cite{baptiste2000}).

Anderson and Potts~\cite{anderson2004} devise a 2-competitive algorithm  for  the online version of $1|r_j|\sum w_jC_j$, i.e., each job $j$ becomes known only at its due-date $r_j$, and scheduling decisions cannot be reversed.  
This is the best possible algorithm for this problem, since Hoogeveen and Vestjens \cite{hoogeveen1996} proved that no online algorithm can have a competitive ratio less than 2, even if the weights of the jobs are identical.

Kellerer et al.~\cite{kellerer1999} describe  a $O(\sqrt{n})$-approximation algorithm for $1|r_j|\sum F_j$ and proved that no polynomial time algorithm can have $O(n^{1/2-\varepsilon})$ approximation ratio for any $\varepsilon>0$ unless $P=NP$.
It is also known that the best competitive ratio of the online problem with unit weights is $\Theta(n)$, and it is unbounded if the jobs have distinct weights \cite{epstein2001}.
Due to these results, most of the research papers assume preemptive jobs, see e.g., \cite{bansal2007,chekuri2001b,epstein2001}.

\new{In one of the online variants of the problem we study in this paper the only unknown parameter is the number of consecutive time periods while the jobs arrive.
The ski-rental problem, introduced by  L.~Rudolph according to Karp~\cite{karp1992line}, is  very much alike. A person would ski for an unknown number of $d$ consecutive days, and she can either rent the equipment for unit cost every day, or buy it in some day for the remaining time for cost $y$.
If she buys the equipment in day $t \leq d$, the total cost is $t+y$, while the offline optimum is $\min\{d,y\}$. One seeks a strategy to minimize the ratio of these two values in the worst case without knowing $d$.
The best deterministic online algorithm has a competitive ratio of 2, while the best randomized online algorithm has $e/(e-1)\approx 1.58$, and both bounds are tight \cite{karlin1988competitive}, \cite{karlin1994competitive}. }

\new{
While this paper is probably the first one to study the joint  replenishment problem combined with production scheduling,
there is a considerable literature on integrated production and outbound distribution scheduling.
In such models, the production and the  delivery of  customer orders are scheduled simultaneously, while minimizing the total production and distribution costs. After the first paper by Potts~\cite{potts1980analysis}, there appeared several models and approaches over the years, for an overview see e.g.~\cite{chen2010integrated}. While most of the papers focus on the off-line variants of the problem, 
there are a few results on the online version as well. In particular, Averbakh et al \cite{averbakh2007line,averbakh2013approximation} propose online algorithms for  single machine problems with linear  competitive ratio in either $\Delta$ or $\rho$, where $\Delta$ and $\rho$  denote the ratio of the maximum to minimum delivery costs, and job processing times, respectively. This is considerably improved by Azar et al.~\cite{azar2016make}, who propose poly-logarithmic competitive online algorithms in $\Delta$ and $\rho$  in single as well as in parallel machine environments, and they also prove that their bounds are essentially best possible.}
\section{Problem statement and overview of the results}\label{sec:prob_statement}
We have a set $\J$ of $n$ jobs that have to be scheduled on a single machine.
Each job $j$ has a processing time $p_j>0$, a release date $r_j\geq 0$, and possibly a weight $w_j>0$ (in case of min-sum type objective functions).
In addition, there is a set of resources $\R=\{R_1,\ldots,R_s\}$, and each job  $j\in\J$ requires a non-empty subset $R(j)$ of $\R$.
Let $\J_i\subseteq \J$ be the set of those jobs that require resource $R_i$.
A job $j$ can only be started if all the resources in $R(j)$ are replenished after $r_j$.
Each time some resource $R_i$ is replenished, a fixed cost $K_i$ is incurred on top of a fixed cost $K_0$, which must be \new{paid} each time moment when any replenishment occurs. These costs are independent of the amount replenished. Replenishment is instantaneous.\footnote{This assumption is without loss of generality in offline problems, and can be handled by transforming the problem data in the online variants, see Section~\ref{sec:online_Cj}.}

A {\em solution of the problem\/} is a pair $(S, \Q)$, where $S$ is a {\em schedule\/} specifying a starting time for each job $j \in \J$, and \new{$\Q=\{(\R_1,t_1),\ldots,(\R_q,t_q)\}$ is a {\em replenishment structure\/}, which specifies time moments $t_\ell$ along with subsets of resources $\R_\ell \subseteq \R$ such that $t_1<\cdots <t_q$.}
We say that job $j$ is \new{\textit{ready to be started at time moment $t$ in replenishment structure $\Q$}}, if each resource $R \in R(j)$ is replenished at some time moment in $[r_j,t]$, i.e., $R(j) \subseteq \bigcup_{t_\ell \in [r_j,t]} \R_\ell$. 
The solution is {\em feasible\/} if (i) the jobs do not overlap in time, i.e., $S_j + p_j \leq S_k$ or $S_k +p_k \leq S_j$ for each $j\neq k$, (ii)  each job $j\in \J$ is \new{ready to be started} at $S_j$ in $\Q$.

The {\em cost of a solution\/} is the sum of the scheduling cost  $c_S$, and the replenishment cost $c_\Q$.
The former can be any optimization criteria know in scheduling theory, but in this paper we confine our discussion to the weighted sum of job completion times $\sum w_j C_j$,  the  sum of flow times $\sum F_j$, where $F_j = C_j - r_j$, and to the maximum flow  time $F_{\max} := \max F_j$.
Note that $C_j$ is the completion time of job $j$ in a schedule.
The replenishment cost is calculated as follows: $c_{\Q} := \sum_{\ell=1}^{|\Q|} (K_0+\sum_{R_i \in \R_\ell} K_i)$.

In the offline problem the complete input is known in advance, and we seek a feasible solution of minimum cost over the set of all feasible solutions. In the online variant of the problem, where the jobs arrive over time, and the input becomes known only gradually, the solution is built step-by-step, however, decisions cannot be reversed, i.e., if a job is started, or some replenishment is made, then this cannot be altered later.

Note that $\sum w_jF_j=\sum w_jC_j-\sum w_jr_j$, where the last sum is a constant, thus the complexity status (polynomial or NP-hard) is the same for the two problems.
However, there can be differences in approximability and the best competitive ratio of algorithms for the online version of the problems.

We extend the notation of \cite{graham1979} by $jrp$ in the $\beta$ field to indicate the joint replenishment of the resources, thus we denote our problem by $1|jrp, r_j|c_S+c_\Q$. In addition, $s=1$ or $s=const$ indicate that the number of resources is  1, and constant, respectively, and not part of the input.
Further on, $p_j = 1$ and $p_j = p$ imply that all the jobs have the same processing time 1 (unit-time jobs), and $p$, respectively.

The next example describe two feasible solutions for a one-resource problem.

\begin{example} Suppose there are 3 jobs, $p_1=4$, $p_2=p_3=1$, $r_1=0$, $r_2=3$, and $r_3=7$.
If there are 3 replenishments from a single resource $R$, i.e., \new{$\Q = ((\{R\},0),(\{R\},3), (\{R\},7))$}, and the starting times of \new{the jobs} are $S_1=0, S_2=4$, and $S_3=7$, then $(S, \Q)$ is a feasible solution.

However, if there are only two replenishments in $\Q'$ at $t_1=3$ and $t_2=7$, then we have to start the jobs later, e.g., if $S'_1=3,S'_2=7$, and $S'_3=8$, then $(S',\Q')$ is feasible.
Observe that in the second solution we have saved the cost of a replenishment ($K_0+K_1$), however, the total completion time of the jobs has increased from 17 to 24.
See Figure~\ref{fig:example1} for illustration.

\begin{figure}
\begin{tikzpicture}
\def\ox{0} 
\def\oy{0} 
\def\ui{1.5}
\def\uii{3.50}
\def\uiii{8}
\coordinate(o) at (\ox,\oy); 
\coordinate(u1) at (\ui,\oy);
\coordinate(u2) at (\uii,\oy);
\coordinate(u3) at (\uiii,\oy);

\tikzstyle{mystyle}=[draw, minimum height=0.5cm,rectangle, inner sep=0pt,font=\scriptsize]

\def\tl{5.0} 
\def\oyi{0}
\draw [-latex](\ox,\oyi) node[above left]{$S$} -- (\ox+\tl,\oyi) node[above,font=\small]{$t$};

\draw[<-] (o) -- ($(o)-(0,0.2)$) node[below] {\tiny 0};
\draw[<-] (u1) -- ($(u1)-(0,0.2)$) node[below] {\tiny 3};
\draw[<-] (u2) -- ($(u2)-(0,0.2)$) node[below] {\tiny 7};

\def\pi{0.7}
\node(b1) [above right=-0.01cm and -0.01cm of o,mystyle, minimum width=2 cm]{$j_1$};
\node(b3) [right=0cm of b1,mystyle, minimum width=0.5cm]{$j_2$};
\node(b3) [above right=-0.01cm and -0.00cm of u2,mystyle, minimum width=0.5cm]{$j_3$};
\end{tikzpicture}
\begin{tikzpicture}
\def\ox{0} 
\def\oy{0} 
\def\ui{1.5}
\def\uii{3.50}
\def\uiii{8}
\coordinate(o) at (\ox,\oy); 
\coordinate(u1) at (\ui,\oy);
\coordinate(u2) at (\uii,\oy);
\coordinate(u3) at (\uiii,\oy);

\tikzstyle{mystyle}=[draw, minimum height=0.5cm,rectangle, inner sep=0pt,font=\scriptsize]

\def\tl{5.0} 
\def\oyi{0}
\draw [-latex](\ox,\oyi) node[above left]{$S'$} -- (\ox+\tl,\oyi) node[above,font=\small]{$t$};

\draw[<-] (u1) -- ($(u1)-(0,0.2)$) node[below] {\tiny 3};
\draw[<-] (u2) -- ($(u2)-(0,0.2)$) node[below] {\tiny 7};

\def\pi{0.7}
\node(b1) [above right=-0.01cm and -0.01cm of u1,mystyle, minimum width=2 cm]{$j_1$};
\node(b3) [right=0cm of b1,mystyle, minimum width=0.5cm]{$j_2$};
\node(b3) [right=0cm of b3,mystyle, minimum width=0.5cm]{$j_3$};
\end{tikzpicture}
\caption{Two feasible solutions. The arrows below the axis denote the replenishments.}\label{fig:example1}
\end{figure}
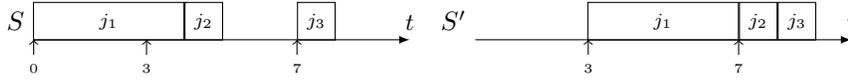

\end{example}

\noindent{\em Results of the paper\/}. The main results of the paper fall in 3 categories: (i) NP-hardness proofs, (ii) polynomial time algorithms, and  (iii) competitive analysis of online variants of the problem, see Table~\ref{tab:results} for an overview. We provide an almost complete complexity classification for the offline problems with both of the $\sum w_j C_j$ and $F_{\max}$ objectives. Notice that the former results imply analogous ones for the $\sum w_j F_j$ criterion. 
While most of our polynomial time algorithms work only with unit-time jobs, a notable exception is the case with a single resource and the $F_{\max}$ objective, where the job processing times are arbitrary positive integer numbers.
We have devised online algorithms for some special cases of the problem for both min-sum and min-max criteria. In all variants for which we present an online algorithm with constant competitive ratio, we  have to assume unit-time jobs. While we have a 2-competitive algorithm with unit time jobs for min-sum criteria, for the online problem with the $F_{\max}$ objective we also have to assume that the input is \new{regular}, i.e., in every time unit a new job arrives, but in this case the competitive ratio is $\sqrt{2}$.
\begin{table}
\caption{Results of the paper.}
\label{tab:results}
\resizebox{\textwidth}{!}{%
\begin{tabular}{ccc}
\hline
Problem&Result&Source\\
\hline
$1|jrp,s=1,r_j|\sum C_j+c_\Q$	&NP-hard	&Proposition~\ref{cor:s1_np} \\
$1|jrp,p_j=1,r_j|\sum C_j+c_\Q$	&NP-hard	&Theorem~\ref{thm:pw1_np} \\
$1|jrp,s=2,r_j|\Fmax+c_\Q$			&NP-hard	&Theorem~\ref{thm:offline_fmax_np} \\
\hline
$1|jrp,s=const,p_j=1,r_j|\sum w_jC_j+c_\Q$	&polynomial alg.	&Theorem~\ref{thm:p1_dyn} \\
$1|jrp,s=const,p_j=p,r_j|\sum C_j+c_\Q$		&polynomial alg.	&Theorem~\ref{thm:w1_dyn} \\
$1|jrp,s=1,r_j|\Fmax+c_\Q$					&polynomial alg.	&Theorem~\ref{thm:offline_fmax_dyn} \\
$1|jrp,s=const,p_j=p,r_j|\Fmax+c_\Q$		&polynomial alg.	&Theorem~\ref{thm:offline_fmax_pjp} \\
\hline
$1|jrp,s=1,p_j=1,r_j|\sum C_j+c_\Q$		&2-competitive	alg.	&Theorem~\ref{thm:onl_2} \\
$1|jrp,s=1,p_j=1,r_j|\sum C_j+c_\Q$		&no $\left(\frac32-\varepsilon\right)$-competitive	alg.	&Theorem~\ref{thm:onl_negative} \\
$1|jrp,s=1,p_j=1,r_j|\sum w_jC_j+c_\Q$	&no $\left(\frac{\sqrt{5}+1}{2}-\varepsilon\right)$-competitive	alg.	&Theorem~\ref{thm:onl_negative_weighted} \\
$1|jrp,s=1,p_j=1,r_j|\sum F_j+c_\Q$		&2-competitive	alg.	&Theorem~\ref{thm:onl_flow} \\
$1|jrp,s=1,p_j=1,r_j|\sum F_j+c_\Q$		&no $\left(\frac32-\varepsilon\right)$-competitive	alg.	&Theorem~\ref{thm:online_sum_fj_3/2} \\
$1|jrp,s=1,p_j=1,\new{regular}\ r_j|\Fmax+c_\Q$		&$\sqrt{2}$-competitive alg.	&Theorem~\ref{thm:online_fmax_sqrt2} \\
$1|jrp,s=1,p_j=1,\new{regular}\ r_j|\Fmax+c_\Q$					&no $\left(\frac43-\varepsilon\right)$-competitive	alg.	&Theorem~\ref{thm:online_fmax_4/3} \\
$1|jrp,s=1,p_j=1,r_j|\Fmax+c_\Q$		&no $\left(\frac{\sqrt{5}+1}{2}-\varepsilon\right)$-competitive	alg.	&Theorem~\ref{thm:onl_fmax_negative_general_case} \\
\hline
\end{tabular}}
\end{table}

\noindent{\em Preliminaries}.
Let $\T:=\{\tau_1,\ldots,\tau_{|\T|}\}$ be the set of different job release dates $r_j$ such that $\tau_1<\ldots<\tau_{|\T|}$.
For technical reasons we introduce $\tau_{|\T|+1}:=\tau_{|\T|}+\sum_{j\in\J} p_j$.

\begin{observation}
For any scheduling criterion, the problem $1|jrp, r_j| c_S + c_\Q$ admits an optimal solution in which all replenishments occur at the release dates of some jobs. 
\label{obs:release}
\end{observation}

\section{Hardness results for min-sum criteria}\label{sec:hardness_sum}

In this section we prove results for $1|jrp,r_j| \sum w_j C_j + c_Q$.
However, these results remain valid, if we replace $\sum w_j C_j$ to $\sum w_jF_j$, since the difference between the two objective functions is a constant ($\sum_{j\in\J} w_j r_j$).

We say that two optimization problems, $A$ and $B$, are \textit{equivalent} if we can get an optimal solution for  problem $A$ in polynomial time by using an oracle which solves $B$ optimally in constant time, and vice versa. Of course, preparing the problem instance for the oracle takes time as well, and it must be polynomial in the size of the instance of problem $A$ ($B$). 

\new{Recall the definition of $\R$ and $\T$ in  Section~\ref{sec:prob_statement}.}

\begin{theorem}
If $|\R|$ and $|\T|$ are constants, then $1|jrp,r_j|\sum w_jC_j+c_\Q$ is equivalent to $1|r_j|\sum w_jC_j$.\label{thm:equiv}
\end{theorem}
\begin{proof}
If $K_0=K_1=\ldots=K_s=0$ then $1|jrp,r_j|\sum w_jC_j+c_\Q$ is exactly the same as $1|r_j|\sum w_jC_j$ (each resource can be replenished no-cost at any time moment).

For the other direction consider an instance $I$ of $1|jrp,r_j|\sum w_jC_j+c_\Q$.
We prove that we can get the optimal solution for $I$ by solving a constant number of instances of $1|r_j|\sum w_jC_j$.
First, we define $\left(2^{|\R|}\right)^{|\T|}$ replenishment structures  for $I$:
\[
W=\{\Q:\Q = \{(\R_1,\tau_1),\ldots,(\R_{|\T|},\tau_{|\T|})), \text{ where } \R_\ell \subseteq \R \}.
\]

By Observation~\ref{obs:release}, there \new{exists} an optimal solution of $1|jrp,r_j|\sum w_jC_j+c_\Q$ with one of the above replenishment structures.

We define an instance $I_{\Q}$ of $1|r_j|\sum w_jC_j$ for each $\Q\in W$. 
The jobs have the same $p_j$, and $w_j$ values as in $I$, the differences are only in their release dates.
For a given $\Q$, let $r'_j:=\min\{t\geq r_j:R(j)\subseteq \cup_{\tau_\ell \in [r_j,t]}\R_\ell\}$ be the release date of $j$ in $I_{\Q}$.
If this value is infinity for any job $j$, then there is no feasible schedule for $I_\Q$, because some resource in $R(j)$ is not replenished at or after $r_j$. 
Observe that, if $S$ is a feasible solution of $I_\Q$, then $(\Q,S)$ is a feasible solution of $I$, and its objective function value is $\sum w_jC_j+c_{\Q}$, where $C_j$ is the completion time of $j$ in $S$.

Thus, if $(\Q^\star,S^\star)$ is an optimal solution of $I$ then $S^\star$ is an optimal solution for $I_{\Q^\star}$ and the optimal solution of $I_{\Q^\star}$ yields an optimal solution for $I$ with $\Q^\star$.
Invoking the oracle for each $I_\Q$ ($\Q\in W$), we can determine an optimal schedule $S^{\Q}$ for each $I^{\Q}$.
From these schedules we can determine the value $\O_\Q:=c_{\Q}+\sum_{j\in\J} w_jC_j^{\Q}$ for each $\Q\in W$, where $C_j^{\Q}$ denotes the completion time of $j$ in $S^{\Q}$.
Let $\tilde{\Q}\in W$ denote a replenishment structure such that $\O_{\tilde{Q}}=\min\{\O_{\Q}:\Q\in W\}$.
Due to our previous observation,  $(\tilde{\Q},S^{\tilde{\Q}})$ is an optimal solution of $I$.
\end{proof}

Note that Theorem~\ref{thm:equiv} remains valid even if we add restrictions on the processing times of the jobs, e.g., $p_j = p$ or $p_j = 1$.

We need the following results of Lenstra et al.~\cite{lenstra1977}.
\begin{theorem}
The problem $1|r_j|\sum C_j$ is NP-hard.
\end{theorem}
The next results is a direct corollary.
\begin{proposition}\label{cor:s1_np}
$1|jrp,s=1,r_j|\sum C_j+c_\Q$ is NP-hard even if $K_0=K_1=0$.
\end{proposition}

We have another negative result when $p_j=1$ and  $s$ is not a constant:

\begin{theorem}\label{thm:pw1_np}
$1|jrp,p_j=1,r_j|\sum C_j+c_\Q$ is NP-hard even if $K_0=0$, and $K_i=1$, $i=1,\ldots,s$.
\end{theorem}

\begin{proof}
We reduce the MAX CLIQUE problem to our scheduling problem.
\vskip 10pt
MAX CLIQUE: Given a graph $G=(V,E)$, and  $k\in\mathbb{Z}$. 
Does there exist a complete subgraph (clique) of $G$ with $k$ nodes?
\vskip 10pt

Consider an instance $I$ of MAX CLIQUE, from which we create an instance $I'$ of $1|jrp,p_j=1,r_j|\sum C_j+c_\Q$.
We define $|V|+M$ jobs (where $M$ is a sufficiently large number), and $s=|E|$ resources in $I'$.
\new{There is a bijection between  the edges in $E$ and the resources, and likewise, there is a bijection between the nodes in $V$ and the first $|V|$ jobs}.
The release date of these jobs is $r'=0$, and such a job $j$ requires a resource $r$ if and only if the node corresponding to job $j$ is an endpoint of the edge associated with resource $r$.
We define $M$ further jobs with a release date $r''=|V|-k$, which require every resource.

We claim that there is a solution for $I$, if and only if $I'$ admits a solution of objective function value of $\O=2|E|-\binom{k}{2}+\binom{M+|V|+1}{2}$.
If there is a $k$-clique $G'=(V',E')$ in $G$ then we introduce two replenishments: at $r'=0$ we replenish every resource that \new{corresponds} to an edge in $E\setminus E'$, and at $r''$, we replenish every resource. 
Observe that the replenishment cost is $2|E|-|E'|=2|E|-\binom{k}{2}$.
Then, consider the following schedule $S$: schedule the jobs that correspond \new{to} the nodes of $V\setminus V'$ and \new{have} a release date 0 from $t=0$ in arbitrary order, after that schedule the remaining jobs in arbitrary order.
It is easy to see that $S$ is feasible since the number of the jobs scheduled in the first step is $|V|-k=r''$.   
The total completion time of the jobs is $\binom{|V|+M+1}{2}$, thus the total cost is $\O$.

Now suppose that there is a solution of $1|jrp,p_j=1,r_j|\sum C_j+c_\Q$ with an objective function value of $\O$.
Observe that each resource must be replenished at $r''$, since the jobs with a release date $r''$ require every resource. 
The cost of this replenishment is $|E|$.
If there is a gap in the schedule before $r''$\new{,} then the total completion time of the jobs is at least $\binom{|V|+M+2}{2}-(|V|-k)=\binom{|V|+M+1}{2}+M+1+k$, which is greater than $\O$ for $M$ sufficiently large.
Hence, there is no gap in the schedule before $r''$.
It means the resources replenished at $r'=0$ are sufficient for $r''=|V|-k$ jobs.
The total completion time of a schedule without any gap is $\binom{|V|+M+1}{2}$, thus  at most $|E|-\binom{k}{2}$ resources can be replenished at $r'=0$, otherwise, the objective function value would exceed $\O$.
This means at least $|V|-k$ nodes have at most $|E|-\binom{k}{2}$ incident edges.
Consider the  remaining at most $k$ nodes of the graph. 
According to the previous observation, there are at least $\binom{k}{2}$ edges among them, which means this subgraph is a $k$-clique.
\end{proof}

\section{Hardness of $1|jrp, r_j| \Fmax + c_\Q$}\label{sec:hardness_fmax}
\new{When the scheduling objective is $\Fmax$}, the problem with two resources is NP-hard, even in a very special case:

\begin{theorem}\label{thm:offline_fmax_np}
If $s=2$, then the problem $1|jrp,r_j|\Fmax+c_{\Q}$ is NP-hard even if $K_0=K_1=0$ and every job requires only one resource.
\end{theorem}
\begin{proof}
We will reduce the NP-hard PARTITION problem to our scheduling problem.
\vskip 10pt
PARTITION: Given a set $A$ of $n$ items with positive integer sizes \new{$a_1,a_2,\ldots,$ $a_{n}$} such that  $B=\sum_{j\in A}a_j /2$ is integral.
Is there a subset $A_1$ of $A$ such that $\sum_{j\in A_1}a_j=B$?
\vskip 10pt
Consider an instance $I$ of PARTITION, and \new{we define the corresponding instance $I'$ of  $1|r_j,jrp|\Fmax+c_{\Q}$} as follows.
There are $n+2$ jobs: $n$ \textit{partition jobs}, such that $p_j=a_j$, $r_j=B$, and $R(j)=R_1$ ($j=1,\ldots,n$).
The other two jobs require resource $R_2$, they have unit processing times, $r_{n+1}:=0$, and $r_{n+2}:=2B$.
Further on, let $K_0 = K_1 := 0$, and $K_2:=B^2$.

We prove that PARTITION instance $I$ has a solution if and only if $I'$ admits a solution with an objective function value of at most $\O:=B^2+2B+1$.

If $I$ has a solution, then consider the following solution $(S,\Q)$ of $I'$.
There are two replenishments: $\Q=((\{R_1\},B),(\{R_2\},2B))$  thus $c_\Q=B^2$. 
Let the partition jobs correspond to the items in $A_1$ be scheduled in $S$ in arbitrary order in $[B,2B]$, while the remaining partition jobs in $[2B+1,3B+1]$. 
Let $S_{n+1}:=2B$, and $S_{n+2}:=3B+1$ (see Figure~\ref{fig:fmax_np} for illustration), therefore $F_j\leq 2B+1$, $j=1,\ldots,n+2$.
Observe that $(S,\Q)$ is feasible, and its objective function value is exactly $\O$.

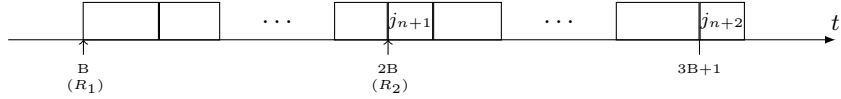
\begin{figure}
\begin{tikzpicture}
\def\ox{0} 
\def\oy{0} 
\def\ui{1.0}
\def\uii{5.05}
\def\uiii{9.19}
\coordinate(o) at (\ox,\oy); 
\coordinate(u1) at (\ui,\oy);
\coordinate(u2) at (\uii,\oy);
\coordinate(u3) at (\uiii,\oy);

\tikzstyle{mystyle}=[draw, minimum height=0.5cm,rectangle, inner sep=0pt,font=\scriptsize]

\def\tl{11.0} 
\def\oyi{0}
\draw [-latex](\ox,\oyi) node[above left]{} -- (\ox+\tl,\oyi) node[above,font=\small]{$t$};

\draw[<-] (u1) -- node[below right= 0.3cm and -0.35 cm]{\tiny ($R_1$)} ($(u1)-(0,0.2)$) node[below] {\tiny B};
\draw[<-] (u2) --node[below right= 0.3cm and -0.35 cm]{\tiny ($R_2$)} ($(u2)-(0,0.2)$) node[below] {\tiny 2B};
\draw[-] (u3) -- ($(u3)-(0,0.2)$) node[below] {\tiny 3B+1};

\def\pi{0.7}
\node(b1) [above right=-0.01cm and -0.01cm of u1,mystyle, minimum width=1.0 cm]{};
\node(b3) [right=0cm of b1,mystyle, minimum width=0.8cm]{};
\node(b3) [right=0cm of b3, minimum width=1.5 cm]{$\ldots$};
\node(b3) [right=0cm of b3,mystyle, minimum width=0.7cm]{};
\node(b3) [right=0cm of b3,mystyle, minimum width=0.5cm]{$j_{n+1}$};
\node(b3) [right=0cm of b3,mystyle, minimum width=0.9cm]{};
\node(b3) [right=0cm of b3, minimum width=1.5 cm]{$\ldots$};
\node(b3) [right=0cm of b3,mystyle, minimum width=1.1cm]{};
\node(b3) [right=0cm of b3,mystyle, minimum width=0.5cm]{$j_{n+2}$};
\end{tikzpicture}
\caption{Illustration of the schedule. The arrows below the axis denote the replenishments (the replenished resources are in a bracket).}\label{fig:fmax_np}
\end{figure}

Now suppose that there is a solution $(S,\Q)$ of $1|jrp,r_j|\Fmax+c_{\Q}$ with $v(S,\Q)\leq\O$.
Observe that there is only one replenishment from $R_2$, since two replenishments would have higher cost than $\O$.
It means the flow time of every job is at most $2B+1$.
Since  $j_{n+2}$ requires $R_2$, the replenishment time of $R_2$ cannot be earlier than $r_{n+2}=2B$.
However, $j_{n+1}$ has to be completed until $2B+1$, which means it requires $R_2$ not later than $2B$, therefore the replenishment of $R_2$ is at $2B$, and $S_{n+1}=2B$.
Since the partition jobs has a flow time of at most $2B+1$, thus these jobs are scheduled in the interval $[B,3B+1]$.  
Due to the position of $j_{n+1}$, the machine cannot be idle in this interval, which means the partition jobs are scheduled in $[B,2B]$, and in $[2B+1,3B+1]$ without idle times ($j_{n+1}$ separates the two intervals). Hence, we get a solution for the PARTITION instance $I$. 
\end{proof}

\section{Polynomial time algorithms for min-sum criteria}\label{sec:dyn_prog}
In this section we describe polynomial time algorithms, based on dynamic programming, for solving special cases of $1|jrp, r_j|\sum w_j C_j + c_Q$. Again, the same methods work for the $\sum w_j F_j + c_Q$ objective function.
Throughout this section we assume that each job $j$ requires a single resource only, and with a slight abuse of notation, let $R(j)$ be this resource.
\new{Recall the definition of $\T$ in the end of Section~\ref{sec:prob_statement}.}

\begin{theorem}\label{thm:p1_dyn}
$1|jrp, s=const,p_j=1,r_j|\sum w_jC_j+c_\Q$ is solvable in polynomial time.
\end{theorem}

\begin{proof}
We describe a dynamic program for computing an optimal solution of the problem.
The states of the dynamic program are arranged in $|\T|+1$ layers, $\States(\tau_1), \ldots, \States(\tau_{|\T|+1})$, where $\States(\tau_k)$ is the \textit{$k^{th}$ layer} ($k=1,\ldots,|\T|+1$). 
A state $N \in \States(\tau_k)$ is a tuple encoding some properties of a partial solution of the problem. Each state has  has a corresponding partial schedule $S^N$ such that the last job scheduled completes not later than $\tau_k$, and also a replenishment structure $\Q^N$ such that \new{replenishments occur only at $\{\tau_1,\ldots,\tau_{k-1}\}$}.
 
\new{More formally, a state $N$ from $\States(\tau_k)$ is a tuple $[\tau_k;\alpha_1,\ldots,\alpha_s;\beta_1,\ldots,\beta_s;$ $ \gamma_1,\ldots,\gamma_s; \delta]$,  with the following definitions of the variables:
\begin{itemize}
\item[]$\alpha_i$ -- number of the jobs from $\J_i$ (jobs that require resource $R_i$) that are scheduled in $S^N$, 
\item[]$\beta_i$ -- time point of the last replenishment from $R_i$ in $\Q^N$ (i.e., $\beta_i$ is the largest time point $\tau_\ell$ such that $R_i\in\R_\ell$, and $(\R_\ell, \tau_\ell) \in \Q^N$, or $'-'$ if there is no such replenishment in $\Q^N$),
\item[]$\gamma_i$ -- number of the replenishments from $R_i$ in $\Q^N$ (the number of $(\R_\ell, \tau_\ell)$ pairs in $\Q^N$ such that $R_i\in\R_\ell$),
\item[]$\delta$ -- total number of replenishments in $\Q^N$, i.e., the number of $(\R_\ell, \tau_\ell)$ pairs in $\Q^N$ such that $\R_\ell \neq \emptyset$. 
\end{itemize} 

See Figure~\ref{fig:dyn_prog} for an easy example.
}
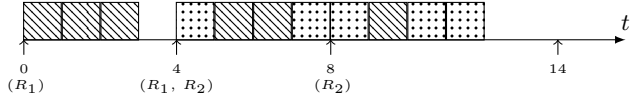
\begin{figure}
\new{
\begin{tikzpicture}
\def\ox{0} 
\def\oy{0} 
\def\ui{0.0}
\def\uii{2.03}
\def\uiiv{4.08}
\def\uiii{7.1}
\coordinate(o) at (\ox,\oy); 
\coordinate(u1) at (\ui,\oy);
\coordinate(u2) at (\uii,\oy);
\coordinate(u25) at (\uiiv,\oy);
\coordinate(u3) at (\uiii,\oy);

\tikzstyle{mystyle}=[draw, minimum height=0.5cm,rectangle, inner sep=0pt,font=\scriptsize,pattern=dots]
\tikzstyle{mystyle2}=[draw, minimum height=0.5cm,rectangle, inner sep=0pt,font=\scriptsize,pattern=north west lines]

\def\tl{8.0} 
\def\oyi{0}
\draw [-latex](\ox,\oyi) node[above left]{} -- (\ox+\tl,\oyi) node[above,font=\small]{$t$};

\draw[<-] (u1) -- node[below right= 0.3cm and -0.35 cm]{\tiny ($R_1$)} ($(u1)-(0,0.2)$) node[below] {\tiny $0$};
\draw[<-] (u2) --node[below=0.3cm]{\tiny ($R_1,\,R_2$)} ($(u2)-(0,0.2)$) node[below] {\tiny $4$};
\draw[<-] (u25) --node[below right= 0.3cm and -0.35 cm]{\tiny ($R_2$)} ($(u25)-(0,0.2)$) node[below] {\tiny $8$};
\draw[<-] (u3) -- ($(u3)-(0,0.2)$) node[below] {\tiny $14$};

\def\pi{0.5}
\node(b1) [above right=-0.01cm and -0.01cm of u1,mystyle2, minimum width=\pi cm]{};
\node(b3) [right=0cm of b1,mystyle2, minimum width=\pi cm]{};
\node(b3) [right=0cm of b3,mystyle2, minimum width=\pi cm]{};

\node(b1) [above right=-0.01cm and -0.01cm of u2,mystyle, minimum width=\pi cm]{};
\node(b3) [right=0cm of b1,mystyle2, minimum width=\pi cm]{};
\node(b3) [right=0cm of b3,mystyle2, minimum width=\pi cm]{};
\node(b3) [right=0cm of b3,mystyle, minimum width=\pi cm]{};

\node(b1) [above right=-0.01cm and -0.01cm of u25,mystyle, minimum width=\pi cm]{};
\node(b3) [right=0cm of b1,mystyle2, minimum width=\pi cm]{};
\node(b3) [right=0cm of b3,mystyle, minimum width=\pi cm]{};
\node(b3) [right=0cm of b3,mystyle, minimum width=\pi cm]{};
\end{tikzpicture}
}
\caption{\new{A corresponding partial schedule and replenishment structure of state $[14;6,5;4,8;2,2;3]$. Jobs of $\J_1$ are hatched, while the jobs of $\J_2$ are dotted.}}\label{fig:dyn_prog}
\end{figure}

The schedule $S^N$ associated with $N$ is specified by $n$  numbers that describe the completion time of the jobs, or $-1$, if the corresponding job is not scheduled in $S^N$.

\new{Let $\J_i(\bar{\beta_i})$ be the subset of those jobs from $\J_i$ that have a release date at most $\bar{\beta}_i$.}

The first layer of the dynamic program has only one state $N_0$ where all parameters are 0 or $'-'$, and the corresponding schedule and replenishment structure is empty.
The algorithm builds the states of $\States(\tau_{k+1})$ from those of $\States(\tau_{k})$, and each state has a 'parent state' in the previous layer (except $N_0$ in $\States(\tau_1)$).
We compute the states by Algorithm~\ref{alg:dyn_prog1}. \new{When processing state $N \in \States(\tau_k)$, the algorithm goes through all the subsets of the resources, and for each subset it generates a tuple $\bar{N}$ by selecting jobs unscheduled in $S^N$.
The number of jobs selected is at most $\tau_{k+1}-\tau_k$, so that they can fit into the interval $[\tau_k, \tau_{k+1}]$. From each set $\J_i$ we consider only  those unscheduled jobs that are ready to be started while taking into account the latest replenishment time of resource $R_i$, see step~\ref{step:choosejobs}. Then, $S^N$ is extended with the selected jobs to obtain a new schedule $S^{new}$, and also a new replenishment structure $Q^{new}$ is obtained in step~\ref{step:construct_schedule}.
There are two cases to consider:
\begin{itemize}
\item If $\bar{N}$ is a new state, i.e., not generated before, then it is unconditionally added to $L(\tau_{k+1})$, see step~\ref{step:new_state}.
\item If $\bar{N}$ is not a new state, then we compare the schedule already associated with it and $S^{new}$. If $S^{new}$ is the better, then $(S^{new},Q^{new})$ replaces $(S^{\bar{N}}, \Q^{\bar{N}})$, and the parent node of $\bar{N}$ is updated to $N$, see step~\ref{step:dyn_last}.
\end{itemize}
}

\begin{varalgorithm}{DynProg{\_}wjCj}
\begin{small}
\caption{}\label{alg:dyn_prog1}
\new{
Initialization: $k:=1$, $N_0:=[\tau_1;0,\ldots,0;-,\ldots,-;0,\ldots,0;0]$, $S^{N_0}$ and $\Q^{N_0}$ are empty sets, $\States(\tau_1) :=\{N_0\}$
\begin{enumerate}
\item For each state $N:=[\tau_k;\alpha_1,\ldots,\alpha_s;\beta_1,\ldots,\beta_s; \gamma_1,\ldots,\gamma_s; \delta]\in \States(\tau_k)$ and for each subset $\bar{\mathcal{R}}$ of $\{R_1,\ldots,R_s\}$ do steps \ref{step:dyn1}-\ref{step:dyn_last}:
\label{step:choose_state_and_R}
\begin{enumerate}[a.]
\item For each resource $R_i$, $i \in \{1,\ldots,s\}$, determine $\bar{\beta}_i$ and $\bar{\gamma}_i$ as follows:
if $R_i\in \bar{\mathcal{R}}$, then $\bar{\beta}_i:=\tau_{k}$ and $\bar{\gamma}_i:=\gamma_i+1$;
otherwise, $\bar{\beta}_i:=\beta_i$ and $\bar{\gamma}_i:=\gamma_i$. \label{step:dyn1} 
\item If $\bar{\mathcal{R}}=\emptyset$, then let $\bar{\delta}:=\delta$, otherwise, let $\bar{\delta}:=\delta+1$.
\label{step:dyn2}
\item 
Determine the sets of jobs $\J_i(\bar{\beta}_i)$ ($i=1,\ldots,s$).
Let $\J^k\subseteq \cup_{i=1}^s \J_i(\bar{\beta}_i)$ be  the $\min\{\tau_{k+1}-\tau_k,\sum_{i=1}^s (|\J_i(\bar{\beta}_i)|- \alpha_i)\}$ unscheduled jobs with the largest weights. \label{step:choosejobs}
\item
Let $\bar{\alpha}_i:=\alpha_i+|\{j\in \J^k: R(j)=i\}|$ for $i=1,\ldots,s$. \label{step:determine_alpha_tau}
Let $\bar{\tau} := \tau_{k+1}$. 
\item Let $\bar{N}:=[\bar{\tau};\bar{\alpha}_1,\ldots,\bar{\alpha}_s;\bar{\beta}_1,\ldots,\bar{\beta}_s;\bar{\gamma}_1,\ldots,\bar{\gamma}_s;\bar{\delta}]$. \label{step:define_N_bar}
\item 
Construct the schedule $S^{new}$ from $S^N$ by scheduling the jobs of $\J^k$ in non-increasing $w_j$ order from $\tau_k$ without idle times. 
Analogously, $\Q^{new}$ is obtained from $\Q^N$ by extending it with $(\bar{\R}, \bar{\tau})$. \label{step:construct_schedule}
\item If 
$\bar{N}\notin \States(\bar{\tau})$, then add $\bar{N}$ to $\States(\bar{\tau})$, set $N$ to be its parent node, and associate $(S^{new}, \Q^{new})$ with  $\bar{N}$.\label{step:new_state}
\item If $\bar{N}\in \States(\bar{\tau})$ and 
the total weighted completion time of $S^{new}$ is smaller than that of $S^{\bar{N}}$, then replace  $S^{\bar{N}}$ with $S^{new}$, $\Q^{\bar{N}}$ with $\Q^{new}$, and set the parent node of $\bar{N}$ to $N$.
\label{step:dyn_last}
\end{enumerate}
\item If $k\leq|\T|$ then let $k:=k+1$ and go to the first step. 
\item  Determine a state $\hat{N} \in \States(\tau_{|\T|+1})$ with the smallest $\sum_{j \in \J} w_j C_j(S^N) + K_0 \delta + \sum_{i=1}^s K_i \gamma_i$ value and output the corresponding solution. \label{step:choose_opt_sol}
\end{enumerate}}
\end{small}
\end{varalgorithm}

Observe that the number of the states is at most $n^{3s+2}$, because there are at most $n$ different options in each coordinate of a state.
This observation shows that the above procedure is polynomial in the size of the input parameters of the problem.

\begin{lemma} \label{lem:feas}
Let $N=[\tau_{k};\alpha_1,\ldots,\alpha_s;\beta_1,\ldots,\beta_s; \gamma_1,\ldots,\gamma_s; \delta]\in \States(\tau_k)$ be an arbitrary state.
\begin{enumerate}[(i)] 
\item 
Each job in $S^N$ completes not later than $\tau_k$.\label{lem:feas_cj} 
\item 
The jobs do not overlap in $S^N$.\label{lem:feas_overlap}
\item 
If $\beta_i$ is a number (i.e., $\beta_i\neq -$), then there is a replenishment from $R_i$ at $\beta_i$ in $\Q^N$.\label{lem:feas_replen}
\item 
There are $\delta$ replenishments in $\Q^N$ of which there are $\gamma_i$ replenishments from $R_i$ ($i=1,\ldots,s$).\label{lem:feas_no_replen}
\item 
If each $\alpha_i=|\mathcal{J}_i|$ then the full schedule $S^N$ is feasible for $\Q^N$.\label{lem:feas_alpha}
\item The cost of the replenishments in $\Q^N$ can be determined from $N$ in polynomial time.\label{lem:feas_Qcost}
\end{enumerate}
\end{lemma}
\begin{proof}
\begin{enumerate}[(i)]
\item 
Follows from  the observation  that $|\J^k|\leq \tau_{k+1}-\tau_k$ in Step~\ref{step:choosejobs}.
\item Follows from (\ref{lem:feas_cj}).
\item Follows from the definition of $\beta_i$.
\item 
Follows from Algorithm \ref{alg:dyn_prog1}.
\item 
Each job is scheduled because $\alpha_i=|\J_i|$ and the jobs do not overlap due to (\ref{lem:feas_overlap}).
Observe that Algorithm \ref{alg:dyn_prog1} cannot schedule a job before its release date. 
If it chooses a job $j\in\J_i$ to process in $[\tau_k,\tau_{k+1}]$ at step \new{\ref{step:choosejobs}}, it has a release date at most $\bar{\beta}_i$, since \new{$j\in\J_i(\bar{\beta}_i)$}, thus the statement follows from (\ref{lem:feas_replen}) \new{and the definition of $\bar{\beta}_i$ in step \ref{step:dyn1}}.
\item
Follows from (\ref{lem:feas_no_replen}). 
The cost of $\Q^N$ is $K_0\cdot \delta + \sum_{i=1}^s K_i\cdot\gamma_i$.
\end{enumerate}
\end{proof}

Consider the set $\J^k$ at Step~\ref{step:choosejobs}, when the algorithm creates an arbitrary state $\bar{N}$ from its parent state $N$. 
\begin{claim}
If $|\J^k|< \tau_{k+1}-\tau_k$, then $\J^k$ is the set of jobs \new{ready to be started} at $\tau_k$ with respect to $\Q^{new}$ that are unscheduled in $S^{N}$.
Otherwise, it is the $\tau_{k+1}-\tau_k$ elements of this set with the largest weights. \label{clm:jk}
\end{claim}
\begin{proof}
Follows from the definitions.
\end{proof}

After determining the states,  Algorithm~\ref{alg:dyn_prog1} chooses a state $\hat{N}\in \States(\tau_{|\T|+1})$ representing a feasible solution of smallest objective function value, and later we prove that $(\Q^{\hat{N}},S^{\hat{N}})$ is an optimal solution.
For each state $N \in \States(\tau_{|\T|+1})$, we decide whether $(\Q^N,S^N)$ is feasible or not (by Lemma~\ref{lem:feas}~(\ref{lem:feas_alpha})), calculate the cost of $\Q^N$ (by Lemma~\ref{lem:feas}~(\ref{lem:feas_Qcost})), and the total weighted completion time of $S^N$.
We define $\hat{N}$ as the state such that $(\Q^{\hat{N}},S^{\hat{N}})$ is feasible, and it has  best objective function value among that of the states of the last layer with a feasible solution.
Observe that this procedure requires polynomial time.

%
%
%
%

Now we prove that there is a state of the last layer such that the corresponding solution of the problem is optimal, thus the solution $(\Q^{\hat{N}},S^{\hat{N}})$ calculated by our method must be optimal.
Let $(\Q^\star,S^\star)$ be an optimal solution.
Notice that the jobs scheduled in the interval $[\tau_k,\tau_{k+1}]$ in $S^\star$ are in non-increasing $w_j$ order, and replenishments occur only at the release dates of the jobs.
Recall that $N_0$ is the only state of $\States(\tau_{1})$.
Let $N_k\in \States(\tau_{k+1})$ be the state that we get from $N_{k-1} \in \States(\tau_k)$, when $\bar{\R}$ was set to \new{$\R^\star_k$ (the set of those resources replenished at $\tau_k$ in $\Q^\star$)} in Algorithm \ref{alg:dyn_prog1} ($k=1,\ldots,|\T|$).
Due to the feasibility of the optimal solution, there is a replenishment from $R_i$ not earlier than $\max_{j\in \J_i}r_j$, thus Algorithm \ref{alg:dyn_prog1} will schedule every job in $S^{N_{|\T|}}$, i.e., each $\alpha_i=|\J_i|$ in $N_{|\T|}$. Therefore, $(\Q^{N_{|\T|}}, S^{N_{|\T|}})$  is feasible by Lemma~\ref{lem:feas}~(\ref{lem:feas_alpha}).

We prove that $(\Q^{N_{|\T|}}, S^{N_{|\T|}})$ is an optimal solution.
Observe that $\Q^\star\equiv \Q^{N_{|\T|}}$ (i.e., the $\Q^\star$ and $\Q^{N_{|\T|}}$ consists of the same pairs of resource subsets and time points).
It \new{remains to prove} that the total weighted completion time is the same for $S^\star$ and $S^{N_{|\T|}}$.
The next claim describes an important observation on these schedules.

\begin{claim}\label{clm:periods}
The machine is working in the same periods in $S^\star$ and in $S^{N_{|\T|}}$. 
\end{claim}
\begin{proof}

Suppose for a contradiction that $\tau_k\leq t< \tau_{k+1}$ is the first time point, when  the machine  becomes idle either in $S^\star$ or in $S^{N_{|\T|}}$, but not in the other.

\new{If} the machine  is idle from $t$ only in $S^\star$, then there are more jobs scheduled until $t+1$ in $S^{N_{|\T|}}$ than in $S^\star$.
Hence, there exists a resource $R_i$ such that there are more jobs scheduled from $\J_i$ (the set of jobs that require $R_i$) until $t+1$ in  $S^{N_{|\T|}}$ than in \new{$S^{\star}$}.
This means there is at least one job $j\in \J_i$ with a release date not later than the last replenishment before $t+1$ from $R_i$,  which starts later than $t$ in $S^\star$, otherwise, $S^{N_{|\T|}}$ would not be feasible.
However, if we modify the starting time of $j$ in $S^\star$ to $t$ (the starting time of  other jobs does not change) then we would get a schedule $S'$ such that $(\Q^\star,S')$ is better than $(\Q^\star,S^\star)$, a contradiction.
 
On the other hand, if  the machine  is idle from $t$ only in $S^{N_{|\T|}}$, then there is a resource $R_i$ such that there are more jobs scheduled from $\J_i$ until $t+1$ in $S^\star$ than in $S^{N_{|\T|}}$. 
The machine is idle from $t$ in $S^{N_{|\T|}}$, thus when Algorithm~\ref{alg:dyn_prog1} creates $N_k$, it chooses a set $\J^k$ at step~\ref{step:choosejobs} such that  $|\J^k|<\tau_{k+1}-\tau_k$.
Hence, we know from Claim~\ref{clm:jk} that each job from \new{$\J^k \cap \J_i(\tau_k)$} is scheduled until $t$ in $S^{N_{k}}$ at that step.
Observe that these jobs are scheduled until $t$ also in $S^{N_{|\T|}}$, and all the other jobs from $\J_i$ must be scheduled after $\tau_{k+1}$ \new{in} any feasible schedule for replenishment structure $\Q^\star$.
It is in a contradiction with the definition of $\J_i$.
\end{proof}

From Claim~\ref{clm:jk} we know that Algorithm~\ref{alg:dyn_prog1} always \new{starts} a job with highest weight among the unscheduled jobs that are \new{ready to be started at any time moment $t$, if any}, during the construction of $S^{N_{|\T|}}$.
We claim that this property is even valid for $S^\star$, thus for every time moment $t$, the job \new{started} at $t$ in $S^{N_{|\T|}}$ has the same weight as the job \new{which starts} at $t$ in $S^\star$, therefore, the total weighted completion time of the two schedules is the same.

Suppose  that the above property \new{is not met by} $S^\star$, that is, let $j_1$ be a job of minimum starting time in $S^\star$ such that there is a job $j_2$ such that $w_{j_2}>w_{j_1}$, $S^\star_{j_2}>S^\star_{j_1}$, and $j_2$ is \new{ready to be started} at $S^\star_{j_1}$ in  $\Q^\star$.
Let $S'$ be the schedule such that $S'_{j_1}:=S^\star_{j_2}$, $S'_{j_2}:=S^\star_{j_1}$, and $S'_j:=S^\star_j$ for each $j\neq j_1,j_2$.
Observe that (i) $j_2$ is \new{ready to be started} at $S'_{j_2}$ in $\Q^\star$ \new{by assumption}, (ii) $j_1$ is \new{ready to be started} at $S'_{j_1}$ in $\Q^\star$, because $r_{j_1}\leq S^\star_{j_1}\leq S'_{j_1}$, and (iii) each other job $j\notin\{j_1,j_2\}$ is ready to be started at $S'_j$ in $\Q^\star$, because $S'_j=S^\star_j$.
Thus schedule $S'$ is feasible for $\Q^\star$. 
The total weighted completion time of $S'$ is smaller than that of $S^\star$, which contradicts that $(\Q^\star,S^\star)$ is optimal. This proves the theorem.
\end{proof}

We can prove Theorem~\ref{thm:w1_dyn} with a slightly modified version of the previous algorithm.

\begin{theorem}\label{thm:w1_dyn}
$1|jrp, s=const,p_j=p,r_j|\sum C_j+c_\Q$ is solvable in polynomial time.
\end{theorem}

\begin{proof}[Proof (sketch)]
We sketch a similar dynamic program to that of Theorem \ref{thm:p1_dyn}.
Let $\T':=\{\tau+ \lambda\cdot p:\tau\in\T, \lambda =0,1,\ldots,n\}=\{\tau'_1<\ldots<\tau'_{|\T'|}\}$.
Observe that $|\T'|\in O(n^2)$, and $\T'$ contains all the possible starting and completion times of the jobs in any schedule without unnecessary idle times.

We modify Algorithm \ref{alg:dyn_prog1} as follows.
There are $|\T'|+1$ layers and each state is of the form $[\tau';\alpha_1,\ldots,\alpha_{s};\beta_1,\ldots,\beta_s; \gamma_1,\ldots,\gamma_s; \delta]$, where $\tau' \in \T'$, and all the states with the same $\tau'$  constitute $\States(\tau')$. The other coordinates have the same meaning as before.

Let $N = [\tau';\alpha_1,\ldots,\alpha_{s};\beta_1,\ldots,\beta_s; \gamma_1,\ldots,\gamma_s; \delta] \in \States(\tau')$ be the state chosen in Step~\ref{step:choose_state_and_R} of the Algorithm \ref{alg:dyn_prog1}. 
Steps \ref{step:dyn1} and \ref{step:dyn2} remain the same and we  modify \new{step~\ref{step:choosejobs} } as follows.
We define $\J^k$ as the set of all unscheduled job from \new{$\bigcup_{i=1}^s\J_i(\bar{\beta}_i)$}. \new{In step~\ref{step:determine_alpha_tau}, we define $\bar{\alpha}_i$ in the same way as in Algorithm \ref{alg:dyn_prog1}}. If $\J^k\neq \emptyset$, then $\bar{\tau}' = \tau' + |\J^k|\cdot p$ (clearly, $\bar{\tau}' \in \T'$ by the definition of $\T'$).
If $\J^k=\emptyset$, then $\bar{\tau}'$ of $\bar{N}$ equals the next member of $\T'$ after $\tau'$ of $N$.
\new{In steps~\ref{step:construct_schedule}, we schedule the jobs  of $\J^k$ in non-increasing release date order. In steps~\ref{step:dyn_last} and \ref{step:choose_opt_sol}, we \new{use unit weights}}.  The remaining steps are the same.

The proof of soundness of the modified dynamic program is analogous to that of  Theorem \ref{thm:p1_dyn}.
\end{proof}
\begin{remark}
The results of Theorems~\ref{thm:p1_dyn} and \ref{thm:w1_dyn} remain valid even if the jobs may require more than one resource. 
\new{To this end, we consider all non-empty subsets of the $s$ resources (there are $2^s-1$ of them), and let $\phi_\ell$ be the $\ell^{th}$ subset in some arbitrary enumeration of them.
Then we define $\J_\ell$ as the set of those jobs that require the subset of resources $\phi_\ell$.
Clearly, the $\J_\ell$ constitute a partitioning of the jobs. In the dynamic program, let $\alpha_\ell$ be the number of jobs scheduled from $\J_\ell$, and this also means that each state has  $2^s-1$ components $\alpha_\ell$. 
In step \ref{step:choosejobs}, $\J^k$ is chosen from
$\cup_{\ell=1}^{2^s-1} \J_\ell(\min_{i\in\phi_\ell} \bar{\beta}_i)$ as the $\min\{\tau_{k+1} - \tau_k, 
\sum_{\ell=1}^{2^s-1} (|\J_\ell(\min_{i\in\phi_\ell} \bar{\beta}_i)|-\alpha_\ell)\}$ unscheduled jobs of largest weight.
With these modifications, the dynamic programs solve the problems, where the jobs may require more than one resource.}
\end{remark}

\section{Polynomial time algorithms for $1|jrp,r_j|\Fmax + c_Q$}
\label{sec:poly_fmax}
In this section we obtain new complexity results for the maximum flow time objective. We start with a special case with a single resource.
\begin{lemma}\label{lem:offline_fmax_struct}
If $s=1$, there is an optimal solution, where jobs are scheduled in non-decreasing $r_j$ order, all replenishments occur at the release dates of some jobs, and the machine is idle only before replenishments.
\end{lemma}
\begin{proof}
Let $(S^\star, \Q^\star)$ be an optimal solution.
Suppose there exist two jobs, $i$ and $j$, such that $r_i < r_j$, and $i$ is scheduled directly after $j$ in $S^\star$. By swapping $i$ and $j$, $F_i$ will decrease, and $F_j$ will not be bigger than $F_i$ before the swap. Hence, $\Fmax$ will not increase, and the  replenishment structure $\Q^\star$ will be feasible for the updated schedule as well, since all replenishments occur at the release dates of some jobs (see Observation~\ref{obs:release}).
 
Let $\tau_1, \tau_2$ be two consecutive replenishment dates in the optimal solution, and assume the machine is idle in the time interval  $[t_1,t_2] \subset [\tau_1,\tau_2)$ such that it is busy throughout the interval $[\tau_1,t_1]$, and let $j_1$ be the first job starting after $t_2$. Shift $j_1$ to the left until $t_1$. This yields a feasible schedule, since again, we may assume that all replenishments occur at the release dates of some jobs. Moreover,  $\Fmax$ will not increase. Repeating this transformation, we can transform $S^\star$ into the desired form.
\end{proof}

\begin{theorem}\label{thm:offline_fmax_dyn}
If $s = 1$, then the problem $1|jrp,r_j|\Fmax+c_\Q$ can be solved in polynomial time.
\end{theorem}

\begin{proof}
Note that if two or more jobs have the same release dates, then in every optimal solution they are replenished at the same time, and they can be scheduled consecutively in arbitrary order without any idle time. Therefore, we can assume that the release dates are distinct, i.e., $r_1 < r_2 < \ldots < r_n$.

Denote with $K=K_0+K_1$ the cost of a replenishment. We define $n$ layers in Algorithm~\ref{alg:offline_fmax_dyn}, layer $i$ corresponding to $r_i$. In layer $i$, the resource is replenished at $r_i$, and the jobs arriving before $r_i$ and having no replenishment yet are scheduled in increasing release date order.

Throughout the algorithm we maintain several states for each layer, where a state $N$ is a 5-tuple, i.e.,  $N = (\alpha_N, \beta_N, \gamma_N, \Fmax(N), u_N)$. 
\begin{itemize}
\item[] $\alpha_N$ -- release date of a job when the last replenishment occurs,
\item[] $\beta_N$ -- release date of a job from which the remaining jobs are scheduled consecutively without idle times,
\item[] $\gamma_N$ -- index of the first job scheduled from $\beta_N$,
\item[] $\Fmax(N)$ -- maximum flow time,
\item[] $u_N$ -- number of replenishments.
\end{itemize}
From each layer $i$, the algorithm derives new states for layers $j=i+1,\ldots,n$. We define an initial layer $0$, consisting of the unique state $N_0 := (0,0,0,0,0)$. 

\medskip
\begin{varalgorithm}{DynProg{\_}Fmax}
\begin{small}
\caption{}\label{alg:offline_fmax_dyn}
Initialization: $N_0 =(0,0,0,0,0)$, $L_0 = \{N_0\}$.
\begin{enumerate}
\item For $i \in \{ 0, \ldots, n \} $ and $j \in \{ i+1, \ldots, n \}$, do steps \ref{step:offline_fmax_dyn_first}-\ref{step:offline_fmax_dyn_last}:
\begin{enumerate}[a.]
\item For each {$N \in L_i$}, $N = (\alpha_{N},\beta_{N},\gamma_{N},\Fmax(N),u_{N})$, define a new state for layer $j$: $N'$. \label{step:offline_fmax_dyn_first}
\item Let $P := \beta_N + \sum_{k=\gamma_N}^{i}{p_k}$ (the completion time of jobs scheduled consecutively).
\item \new{Start} the jobs between $i+1$ and $j$ in increasing release date order from $\max{(P,r_j)}$ consecutively, without idle times.
\item Calculate $F_k = C_k - r_k$ for each newly scheduled job.
\item 
$\alpha_{N'} := r_i$.

$\beta_{N'} :=
  \begin{cases}
    r_j  &  \text{if } P < r_j \\
    \beta_N  &  \text{if } P \geq r_j.
  \end{cases}
$

$\gamma_{N'} :=
  \begin{cases}
    i+1  &  \text{if } P < r_j \\
    \gamma_N  &  \text{if } P \geq r_j.
  \end{cases}
$

$\Fmax(N') := \max(\Fmax(N), \max_{k=i+1}^{j}F_k)$.

$u_{N'} := u_N + 1$.
\item Let $N' = (\alpha_{N'},\beta_{N'},\gamma_{N'},\Fmax(N'),u_{N'})$.
\item Insert $N'$ in $L_j$ if it has no other state with the same  $(\alpha, \beta, \gamma, u)$ attributes, or replace the state $N'' \in L_j$ with $N'$ if $(\alpha_{N'}, \beta_{N'}, \gamma_{N'}, u_{N'}) = (\alpha_{N''}, \beta_{N''}, \gamma_{N''}, u_{N''})$, while $\Fmax(N') < \Fmax(N'')$. \label{step:offline_fmax_dyn_last}
\end{enumerate}
\item Determine the state $N \in L_n$ with the smallest $F_{\max}(N) + K u_N$ value and output the corresponding solution. 
\end{enumerate}
\end{small}
\end{varalgorithm}

From Lemma~\ref{lem:offline_fmax_struct} we know that there is an optimal solution, where the jobs are scheduled in increasing release date order.  In the proposed algorithm we consider all of the replenishment structures, and the best schedule for each of them. Hence,  the algorithm finds the optimal solution.

Since for each $N \in L_i$, $(\alpha_N, \beta_N, \gamma_N, u_N) \in \{r_1,\ldots,r_i\}^2 \times \{1,\ldots,i\}^2$, there are at most $n^4$ states in the layer $n$. Therefore, our algorithm is of polynomial time complexity, and it finds an optimal state $N^\star=(\alpha^\star, \beta^\star, \gamma^\star, F^\star_{\max}, u^\star)$ with objective $F^\star_{\max} + K u^\star$. 

By recording the parent state of each state $N$ in the algorithm, we can define a series of states leading to $N^\star$: $\{N_0,N^\star_1, \ldots, N^\star_k\}$, where $N^\star_k = N^\star$. The replenishment times are defined by $\alpha_1^\star, \ldots \alpha_k^\star$. The optimal schedule can be built in the following way: for $i = k, \ldots, 1$, jobs with indices at least $\gamma^\star_i$ are scheduled from $\beta^\star_i$.
\end{proof}

Now we propose a more efficient algorithm for the case, where $s=1$, $p_j = 1$ for every job, and all the job release dates are distinct (denoted by '{\em distinct $r_j$\/}').
\begin{lemma}\label{lem:offline_fmax_same_flowtime}
The problem $1|jrp, s=1, p_j = 1, \textit{distinct } r_j| \Fmax + c_Q$ admits  an optimal solution, where every job has the same flow time.
\end{lemma}

\begin{proof}
Consider an optimal solution, where there is a job $j$ with $F_j < \Fmax$, directly followed by a job $k$ with $F_k = \Fmax$. Since $r_j < r_k$ and $F_j < F_k$, we have $C_k - C_j = r_k - r_j + \Delta$, where $\Delta > 0$. This also means that there is a gap of length $r_k - r_j + \Delta - 1 > \Delta$ between $j$ and $k$. By moving $j$ to the right with $\Delta$ time units, we obtain a feasible schedule such that $F_j = F_k = \Fmax$. We can repeat this transformation until we obtain an optimal schedule of the desired structure.
\end{proof}

\begin{theorem}\label{thm:offline_fmax_fast}
There is an $O(n^2)$ time algorithm that finds the optimal solution for $1|jrp, s=1, p_j = 1, \textit{distinct } r_j| \Fmax + c_Q$.
\end{theorem}

\begin{proof}

Using Lemma~\ref{lem:offline_fmax_same_flowtime}, if we try every possible flow time for scheduling the jobs, and then find the best replenishment structure for this fixed schedule, we obtain the optimal solution by picking the solution that gives the minimal objective function value. Consider the following algorithm:

\medskip
\begin{varalgorithm}{Distinct{\_}rj{\_}Fmax}
\begin{small}
\caption{}\label{alg:offline_fmax_n2}
\begin{enumerate}
\item For each $F \in \{ 1, \ldots, n \}$, do steps \ref{step:offline_fmax_n2_first}-\ref{step:offline_fmax_n2_last}:
\begin{enumerate}[a.]
\item Construct a feasible schedule, where every job has flow time of $F$, by starting every job $j$ at $r_j +F-1$. Let $u_F := 0$. \label{step:offline_fmax_n2_first}
\item While there is a job without a resource ordered for it, let $j$ be the last such job. Replenish the resource at $r_j$, let $u_F := u_F + 1$, and repeat. \label{step:offline_fmax_n2_repl}
\item The solution for $F$ has objective function value $F + K u_F$.
\label{step:offline_fmax_n2_last}
\end{enumerate}
\item Output the best solution obtained in the first step.\end{enumerate}
\end{small}
\end{varalgorithm}
An example for six jobs with $r_1 = 1, r_2 = 2, r_3 = 3, r_4 = 5, r_5 = 7, r_6 = 8$ and fixed $F = 3$ is shown on Figure~\ref{fig:offline_fmax}.

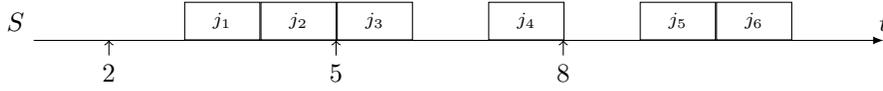
\begin{figure}
\begin{tikzpicture}
\def\ox{0} 
\def\oy{0} 
\def\ui{1}
\def\uii{4.02}
\def\uiii{7.04}

\coordinate(o) at (\ox,\oy); 
\coordinate(u1) at (\ui,\oy);
\coordinate(u2) at (\uii,\oy);

\tikzstyle{mystyle}=[draw, minimum height=0.5cm,rectangle, inner sep=0pt,font=\scriptsize]

\def\tl{11.3} 
\def\oyi{0}
\draw [-latex](\ox,\oyi) node[above left]{$S$} -- (\ox+\tl,\oyi) node[above,font=\small]{$t$};

\coordinate (uq1) at (\ui,\oyi);
\coordinate (uq2) at (\uii,\oyi);
\coordinate (uq3) at (\uiii,\oyi);

\draw[<-] (uq1) -- ($(uq1)-(0,0.2)$) node[below] {$2$};
\draw[<-] (uq2) -- ($(uq2)-(0,0.2)$) node[below] {$5$};
\draw[<-] (uq3) -- ($(uq3)-(0,0.2)$) node[below] {$8$};

\node(b1) [above right=-0.0cm and 1.0cm of u1,mystyle, minimum width=1 cm]{$j_1$};
\node(b2) [right=0.0cm of b1,mystyle, minimum width=1 cm]{$j_2$};
\node(b3) [right=0.0cm of b2,mystyle, minimum width=1 cm]{$j_3$};
\node(b4) [right=1cm of b3,mystyle, minimum width=1 cm]{$j_4$};
\node(b5) [right=1cm of b4,mystyle, minimum width=1 cm]{$j_5$};
\node(b6) [right=0cm of b5,mystyle, minimum width=1 cm]{$j_6$};

\end{tikzpicture}
\caption{Schedule and replenishment structure created by Algorithm~\ref{alg:offline_fmax_n2} for six jobs for $F=3$.}\label{fig:offline_fmax}
\end{figure}

The running time of the algorithm is $O(n^2)$.

We only have to show that for a fixed $F$, the replenishment structure defined by the algorithm gives a solution with a minimal number of replenishments. We proceed by induction on the number of jobs. For $n=1$ the statement clearly holds.
Consider an input $I$ with $n$ jobs, and denote the optimal number of replenishments for $I$ for the fixed flow time $F$ with $u_{F}(I)$.

For the last job in the schedule, there must be a replenishment at $r_n$. 
Delete the jobs which start at or after $r_n$ in the schedule, and let $I'$ be the input obtained  by deleting the same jobs from $I$. By the induction hypothesis, the algorithm determines the optimal number of replenishments for $I'$. Hence, $u_{F}(I) \leq u_{F}(I') + 1$. On the other hand, by the construction of $I'$, for all $j \in I'$, $r_j +F \leq r_n$. Hence, $u_{F}(I) \geq u_{F}(I') + 1$, and thus $u_{F}(I) = u_{F}(I') + 1$.
\end{proof}

Finally, we mention that it is easy to modify the dynamic program of Section~\ref{sec:dyn_prog} for the $\Fmax$ objective, so that we have:

\begin{theorem}\label{thm:offline_fmax_pjp}
$1|jrp, s=const,p_j=p,r_j|\Fmax+c_\Q$ is solvable in polynomial time.
\end{theorem}

\begin{proof}[Proof (sketch)]

\new{

Use the same  dynamic program as in the proof of Theorem~\ref{thm:w1_dyn}, except that in steps~\ref{step:dyn_last} and \ref{step:choose_opt_sol}, consider the maximum flow time instead of the total completion time. 

The proof of soundness of this algorithm is analogous to those of Theorems \ref{thm:p1_dyn} and  \ref{thm:w1_dyn}.
}
\end{proof}

\section{Competitive analysis of the online problem with min-sum criteria}
\label{sec:online_min-sum}
First, we provide a 2-competitive algorithm for the case \new{where the scheduling cost is the total completion time ($\sum C_j$)}, then we improve this algorithm to achieve the same result for the case \new{where the scheduling cost is the total flow time ($\sum F_j$)}.
The proof of the second result requires more sophisticated analysis, but the main idea is the same as that of the proof of the first result.

\subsection{Online algorithm for $1|jrp, s=1, p_j=1, r_j|\sum C_j + c_Q$}
\label{sec:online_Cj}
In the online version of the problem, we do not know the number of the jobs or anything about them before their release dates. 

Since there is only one resource and the processing time and the weight of each job is one, we can suppose that the order of the jobs is the same in any schedule (a fixed non-decreasing release date order).
To simplify the notation, we introduce $K:=K_0+K_1$ for the cost of a replenishment.

At each time point $t$, first we can examine the jobs released at $t$, then we have to decide whether to replenish \new{the resource} or not, and finally we can \new{start} a job from $t$ for which the resource is replenished.
Note that if  we have to decide about the replenishment at time point $t$ \textit{before} we get information about the newly released jobs, then the problem is the same as the previous one with each job release date increased by one, by assuming that ordering takes one time unit.

\new{In the following algorithm, let $\js_t\subseteq\J$ denote the set of unscheduled jobs at time $t$.
When we say that 'start the jobs of $\js_t$  from time $t$', then it means that the jobs in $\js_t$ are put on the machine from time $t$ on without any delays between them in increasing release date order. Observe that their total completion time will be $t\cdot |\js_t| + G(|\js_t|)$, where $G(a) := a(a+1)/2$.}

\medskip
\medskip
\begin{varalgorithm}{Online{\_}SumCj}
\begin{small}
\caption{}\label{alg:online}

Initialization: $t:=0$
\begin{enumerate}
\item \new{Determine the set $\js_t$ of unscheduled jobs at time $t$.}\label{step:onl1_1}
\item \new{If $t\cdot |\js_t|+G(|\js_t|)\geq K$, then replenish the resource, start the jobs of $\js_t$ from $t$, and let $t:=t+|\js_t|$. 
Otherwise, $t := t+1$}.
\label{step:online_sumwork}

\item Go to step~\ref{step:onl1_1}. \label{step:online_last}
\end{enumerate}
\end{small}
\end{varalgorithm}

Let $(S,\Q)$ denote the solution created by Algorithm~\ref{alg:online}, while $(S^\star,\Q^\star)$ an arbitrary optimal solution.
\new{Recall  the notation $v(sol)$ denoting the objective function value of a solution $sol$.}
Let $t_i$ be the time moment of the $i^{th}$ replenishment in $\Q$ and $t_0:=0$.
\new{To simplify our notations, we will use $\js_{(i)}:=\js_{t_i}$ for the set \new{(block)} of jobs  that start} in $[t_i,t_{i+1})$ in $S$, see Figure~\ref{fig:online_S}.

\begin{figure}
\begin{tikzpicture}
\def\ox{0} 
\def\oy{0} 
\def\ui{1}
\def\uii{6.50}
\def\uiii{8}
\coordinate(o) at (\ox,\oy); 
\coordinate(u1) at (\ui,\oy);
\coordinate(u2) at (\uii,\oy);
\coordinate(u3) at (\uiii,\oy);

\tikzstyle{mystyle}=[draw, minimum height=0.5cm,rectangle, inner sep=0pt,font=\scriptsize]

\def\tl{11.0} 
\def\oyi{0}
\draw [-latex](\ox,\oyi) node[above left]{$S$} -- (\ox+\tl,\oyi) node[above,font=\small]{$t$};

\coordinate (uq1) at (\ui,\oyi);
\coordinate (uq2) at (\uii,\oyi);
\coordinate (uq3) at (\uiii,\oyi);
\draw[<-] (uq1) -- ($(uq1)-(0,0.2)$) node[below] {$t_{1}$};
\draw[<-] ($(uq1)-(-1.5,0.0)$) -- ($(uq1)-(-1.5,0.2)$) node[below] {$t_2=t_1+\new{b}_1$};
\draw[<-] ($(uq1)-(-2.7,0.0)$) -- ($(uq1)-(-2.7,0.2)$) node[below right=0cm and -0.2cm] {$t_3=t_2+\new{b}_2$};
\draw[<-] (uq2) -- ($(uq2)-(0,0.2)$) node[below] {$t_4$};
\draw[<-] (uq3) -- ($(uq3)-(0,0.2)$) node[below] {$t_5$};
\draw[<-] ($(uq3)-(-1,0.0)$) -- ($(uq3)-(-1,0.2)$) node[below right=0cm and -0.2cm] {$t_6=t_5+\new{b}_5$};

\def\pi{0.7}
\node(b1) [above right=-0.01cm and -0.01cm of u1,mystyle, minimum width=1.5 cm]{$\js_{(1)}$};
\node(b3) [right=0cm of b1,mystyle, minimum width=1.2cm]{$\js_{(2)}$};
\node(b3) [right=0cm of b3,mystyle, minimum width=1.4 cm]{$\js_{(3)}$};

\node(b3) [above right=-0.01cm and -0.00cm of uq2,mystyle, minimum width=1.2cm]{$\js_{(4)}$};

\node(b3) [above right=-0.01cm and -0.00cm of uq3,mystyle, minimum width=1cm]{$\js_{(5)}$};
\node(b3) [right=0cm of b3,mystyle, minimum width=1.4 cm]{$\js_{(6)}$};

\end{tikzpicture}
\caption{Schedule $S$ created by Algorithm~\ref{alg:online}.}\label{fig:online_S}
\end{figure}
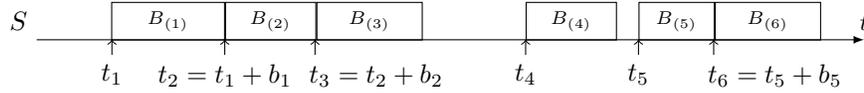

Clearly, the release date $r_j$ of a job $j\in \new{\jsi}$ has to be in $[t_{i-1}+1,t_i]$.
For technical reasons we introduce $\new{b_i:=|\jsi|}$,  $z_i:=|\{j\in \new{\jsi}:r_j=t_i\}|$, and $y_i:=\new{b}_i-z_i$.
The next observation follows from the condition of step~\ref{step:online_sumwork} of Algorithm~\ref{alg:online}.

\begin{observation}\label{obs:onl_schrule}
If the machine is idle in $[t_i-1,t_i]$, then $y_i(t_i-1)+G(y_i)<K$.
\end{observation}

\new{We} divide  $v(S,\Q)$ among the \new{blocks} $\jsi$ in the following way: for $i=1,2,\ldots$, let 
\[
ALG_i:=K+\sum_{j\in\jsi}C_j=K+t_i \new{b}_i+ \new{G(b_i)},
\]
i.e., the total completion time of the jobs of $\jsi$ in $S$ plus $K$, which is the cost of the replenishment at $t_i$.
Since the sets $\jsi$ are disjoint, we have $v(S,\Q)=\sum_{i\geq 1} ALG_i$.
Finally, note that any gap in $S$ has to be finished at $t_i$ for some $i \geq 1$. 
We divide the optimum value in a similar way: we introduce values $OPT_i$, $i=1,2\ldots$, where 
\begin{align*}
OPT_i:=\begin{cases}
K+\sum_{j\in\jsi}C^\star_j, &\text{if  $\exists\ (\R^\star_\ell, \tau_\ell)\in \Q^\star$ such that $t_{i-1} < \tau_\ell < t_i$}.\\
\sum_{j\in\jsi}C^\star_j, &\text{otherwise.}
\end{cases}
\end{align*}
Observe that  $v(S^\star,\Q^\star)\geq\sum_{i\geq 1} OPT_i$.
Now we prove that Algorithm~\ref{alg:online} is 2-competitive.

\begin{theorem}\label{thm:onl_2}
Algorithm~\ref{alg:online} is 2-competitive for the online  \linebreak $1|jrp,s=1,  p_j=1, r_j|\sum C_j+c_\Q$ problem.
\end{theorem}

\begin{proof}

Suppose that there is a gap in $S$ before $t_i$ ($i\geq 1$), and the next gap starts in $[t_{i+\ell},t_{i+\ell+1})$.
Then, we have $\sum_{\mu=0}^{\ell} \new{b}_{i+\mu}$ jobs \new{that start} between the two neighboring gaps, and their completion times are $t_i+\nu$, where $\nu=1,2,\ldots,\sum_{\mu=0}^{\ell} \new{b}_{i+\mu}$.
Hence,
\begin{align*}
\sum_{\mu=0}^{\ell} ALG_{i+\mu}=(\ell+1)K+t_i\sum_{\mu=0}^{\ell} \new{b}_{i+\mu} + \new{G\left(\sum_{\mu=0}^{\ell} b_{i+\mu}\right)}.
\end{align*}

Now we verify that $\sum_{\mu=0}^{\ell} ALG_{i+\mu}\leq 2\cdot \sum_{\mu=0}^{\ell} OPT_{i+\mu}$, thus the theorem follows from the previous observations on $ALG_i$ and $OPT_i$.
Let $0\leq \ell'\leq \ell$ be the smallest index such that  there is no replenishment in $[t_{i+\ell'-1}+1,t_{i+\ell'}-1]$ in $\Q^\star$.
If there is no such index, then let $\ell':=\ell+1$.


\begin{claim}\label{clm:onl_2comp}
$ALG_{i+\mu}\leq 2\cdot OPT_{i+\mu}$ for all $\mu\in[\ell',\ell]$.
\end{claim}
\begin{proof}
If $\ell'=\ell+1$, then the claim is trivial.

Otherwise, we have $r_j\geq t_{i+\ell'-1}+1=:t'$ for each job $j\in\cup_{\mu=\ell'}^{\ell}\js_{(i+\mu)}$, hence they cannot be \new{started} before the first replenishment after $t'$.
The first replenishment after $t'$ in $\Q^\star$ is not earlier than $t_{i+\ell'}$ in $S^\star$ due to the definition of $\ell'$.
We have assumed that the order of the jobs is the same in every schedule, and  there is no idle time among the jobs in $\cup_{\mu=\ell'}^{\ell}\js_{(i+\mu)}$ in $S$ by the choice of $i$ and $\ell$, thus, we have $C^\star_j\geq C_j$ for each $j\in \cup_{\mu=\ell'}^{\ell}\js_{(i+\mu)}$.
Since $OPT_{i+\mu}\geq\sum_{j\in\js_{(i+\mu})}C^\star_j\geq \sum_{j\in\js_{(i+\mu})}C_j\geq K$ for all $\mu\in[\ell',\ell]$, where the last inequality follows from the condition of step~\ref{step:online_sumwork} of Algorithm~\ref{alg:online}, we have $ALG_{i+\mu}=K+\sum_{j\in\js_{(i+\mu)}}C_j\leq 2\cdot OPT_{i+\mu}$ for all $\mu\in[\ell',\ell]$.
\end{proof}

If $\ell'=0$, then we have proved the theorem.
From now on, we suppose $\ell'\geq 1$.
Claim~\ref{clm:onl_2comp} shows that it \new{remains} to prove $\sum_{\mu=0}^{\ell'-1} ALG_{i+\mu}\leq 2\cdot \sum_{\mu=0}^{\ell'-1} OPT_{i+\mu}$, if there is a replenishment in every interval $[t_{i+\mu-1}+1,t_{i+\mu}-1]$, $\mu=0,\ldots,\ell'-1$ in $\Q^\star$.
See Figure~\ref{fig:onl_proof} for an illustration of a possible realization of $S^\star$.

\begin{figure}
\new{
\begin{tikzpicture}
\def\ox{0} 
\def\oy{0} 
\def\ui{1}
\def\uii{6.50}
\def\uiii{8}
\coordinate(o) at (\ox,\oy); 
\coordinate(u1) at (\ui,\oy);
\coordinate(u2) at (\uii,\oy);
\coordinate(u3) at (\uiii,\oy);

\tikzstyle{mystyle}=[draw, minimum height=0.5cm,rectangle, inner sep=0pt,font=\scriptsize]

\def\tl{11.0} 
\def\oyi{0}
\draw [-latex](\ox,\oyi) node[above left]{$S^{\;\;}$} -- (\ox+\tl,\oyi) node[above,font=\small]{$t$};

\coordinate (uq1) at (\ui,\oyi);
\coordinate (uq2) at (\uii,\oyi);
\coordinate (uq3) at (\uiii,\oyi);
\draw[<-] ($(uq1)-(0.2,0.0)$) -- ($(uq1)-(0.2,0.2)$) node[below] {$t_{i-1}$};
\draw[<-] ($(uq1)-(-1.32,0.0)$) -- ($(uq1)-(-1.32,0.2)$) node[below] {$t_i$};
\draw[<-] ($(uq1)-(-3.22,0.0)$) -- ($(uq1)-(-3.22,0.2)$) node[below right=0cm and -0.2cm] {$t_{i+1}=t_i+b_i$};

\draw[dashed] ($(uq1)+(2.9,0.0)$) -- ($(uq1)+(2.9,0.5)$);
\draw[dashed] ($(uq1)+(4.9,0.0)$) -- ($(uq1)+(4.9,0.5)$);

\draw [decorate,decoration={brace,amplitude=3pt}] (2.32,0.5) --node [above]{$y_i$} (3.9,0.5);
\draw [decorate,decoration={brace,amplitude=3pt}] (3.9,0.5) --node [above]{$z_i$} (4.22,0.5);
\draw [decorate,decoration={brace,amplitude=3pt}] (4.22,0.5) --node [above]{$y_{i+1}$} (5.9,0.5);
\draw [decorate,decoration={brace,amplitude=3pt}] (5.9,0.5) --node [above]{$z_{i+1}$} (6.52,0.5);

\def\pi{0.7}
\node(b1) [above right=-0.0cm and -0.2cm of u1,mystyle, minimum width=1.1 cm]{$\js_{(i-1)}$};
\node(b0) [left=0.2cm of b1, minimum width=0.1 cm]{$\ldots$};
\node(b3) [right=0.4cm of b1,mystyle, minimum width=1.9cm]{$\js_{(i)}$};
\node(b3) [right=0cm of b3,mystyle, minimum width=2.3 cm]{$\js_{(i+1)}$};
\node(b3) [right=0.5cm of b3, minimum width=0.1 cm]{$\ldots$};

\end{tikzpicture}

\begin{tikzpicture}
\def\ox{0} 
\def\oy{0} 
\def\ui{1}
\def\uii{7.10}
\def\uiii{8}
\coordinate(o) at (\ox,\oy); 
\coordinate(u1) at (\ui,\oy);
\coordinate(u2) at (\uii,\oy);
\coordinate(u3) at (\uiii,\oy);

\tikzstyle{mystyle}=[draw, minimum height=0.5cm,rectangle, inner sep=0pt,font=\scriptsize]

\def\tl{11.0} 
\def\oyi{0}
\draw [-latex](\ox,\oyi) node[above left]{$S^\star$} -- (\ox+\tl,\oyi) node[above,font=\small]{$t$};

\coordinate (uq1) at (\ui,\oyi);
\coordinate (uq2) at (\uii,\oyi);
\coordinate (uq3) at (\uiii,\oyi);
\draw[] ($(uq1)-(0.2,0.0)$) -- ($(uq1)-(0.2,0.2)$) node[below] {$t_{i-1}$};
\draw[] ($(uq1)-(-1.32,0.0)$) -- ($(uq1)-(-1.32,0.2)$) node[below] {$t_i$};
\draw[] ($(uq1)-(-3.22,0.0)$) -- ($(uq1)-(-3.22,0.2)$) node[below] {$t_{i+1}$};
\draw[] (uq2) -- ($(uq2)-(0,0.2)$) node[below] {$t_{i+\ell'-1}$} node[right=0.15cm]{(no repl.)};
\draw[<-] ($(uq2)+(1.7,0.0)$) -- ($(uq2)+(1.7,-0.2)$) node[below right=0cm and -0.2cm] {$t_{i+\ell'}$};
\draw[<-] ($(uq2)+(2.9,0.0)$) -- ($(uq2)+(2.9,-0.2)$) node[below right=0cm and -0.2cm] {$t_{i+\ell'+1}$};

\draw[<-] ($(uq1)+(0.2,0.0)$) -- ($(uq1)-(-0.2,0.2)$) node[below] {};
\draw[<-] ($(uq1)+(1.5,0.0)$) -- ($(uq1)-(-1.5,0.2)$) node[below] {};
\draw[<-] ($(uq1)+(3.6,0.0)$) -- ($(uq1)-(-3.6,0.2)$) node[below] {};

\draw[dashed] ($(uq1)+(1.8,0.0)$) -- ($(uq1)+(1.8,0.5)$);
\draw[dashed] ($(uq1)+(3.8,0.0)$) -- ($(uq1)+(3.8,0.5)$);

\def\pi{0.7}
\node(b1) [above right=-0.0cm and 0.2cm of u1,mystyle, minimum width=1.9 cm]{$\jsi$};
\node(b0) [left=0.4cm of b1, minimum width=0.1 cm]{$\ldots$};

\node(b3) [right=0cm of b1,mystyle, minimum width=2.3cm]{$\js_{(i+1)}$};
\node(b3) [right=0.5cm of b3, minimum width=0.1 cm]{$\ldots$};

\node(b3) [above right=-0.01cm and 1.70cm of uq2,mystyle, minimum width=1.2cm]{$\js_{(i+\ell')}$};
\node(b3) [right=0.1cm of b3, minimum width=0.1 cm]{$\ldots$};
\end{tikzpicture}}
\caption{\new{Schedule $S$ created by Algorithm~\ref{alg:online} and} a possible realization of $S^\star$.}\label{fig:onl_proof}
\end{figure}
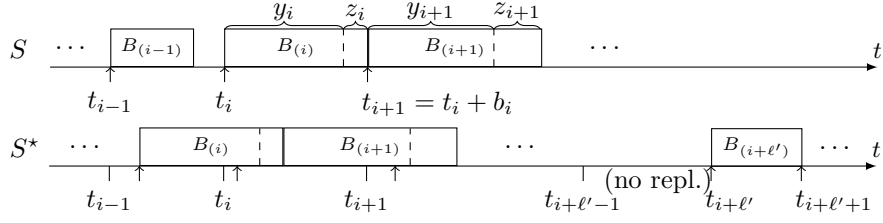

\new{Let $\beta:=\sum_{\mu=1}^{\ell'-1} \new{b}_{i+\mu}$.} 
Observe that 
\begin{align}
\sum_{\mu=0}^{\ell'-1} OPT_{i+\mu}\geq&\ell'K+(t_{i-1}+1)y_i+\new{G(y_i)}+t_i\left(z_i+\new{\beta}\right)+\new{G(z_i + \beta)},\label{ieq:opt_lb1}
\end{align}
since $OPT_{i+\mu}=K+\sum_{j\in\J^S_{i+\mu}}C^\star_j$ for all $0\leq \mu\leq \ell'-1$,  there are $y_i$ jobs with a release date at least $t_{i-1}+1$, and another $z_i+\new{\beta}$ jobs with a release date at least $t_i$. 

\new{On the other hand, we have 
\begin{equation}
\sum_{\mu=0}^{\ell'-1} ALG_{i+\mu}=\ell'K+t_i(y_i+z_i+\beta)+G(y_i+z_i+\beta),\label{ieq:alg_form}
\end{equation}
}
therefore,
\begin{align*}
2\cdot\sum_{\mu=0}^{\ell'-1} OPT_{i+\mu}&-\sum_{\mu=0}^{\ell'-1} ALG_{i+\mu}\\
\geq & \ell'K+2(t_{i-1}+1)y_i+t_i\left(z_i+\new{\beta}\right)-t_iy_i+\new{2G(y_i)}+\\
&\new{2G(z_i+\beta)-G(y_i+z_i+\beta)}\\
\geq &\ell'K+(2t_{i-1}+1)y_i+t_i\left(z_i+\new{\beta}\right)-(t_i-1)y_i\\
\geq &\ell'K+(2t_{i-1}+1)y_i+t_i\left(z_i+\new{\beta}\right)-K\geq 0,
\end{align*}
where \new{the first inequality is a direct consequence of (\ref{ieq:opt_lb1}) and (\ref{ieq:alg_form}), the second one follows from the fact that $2G(a) + 2G(b) \geq G(a+b)$ for any $a$ and $b$.} 
The first inequality of the last line follows from Observation~\ref{obs:onl_schrule}, and the second from $\ell'\geq 1$.

Finally, we show that the above analysis is tight.
Suppose that there is only one job with a release date 0.
Algorithm~\ref{alg:online} \new{starts} this job from $t=K-1$, thus $v(S,\Q)=2K$.
In the optimal solution $(S^\star,\Q^\star)$, this job starts at $t=0$, thus $v(S^\star,\Q^\star)=K+1$.
As $K$ tends to infinity $2K/(K+1)$ tends to 2, which shows that  Algorithm~\ref{alg:online} is not $\alpha$-competitive  for any $\alpha<2$.
\end{proof}
We close this section by lower bounds on the best possible competitive ratios:

\begin{theorem}\label{thm:onl_negative}
There is no $\left(\frac32-\varepsilon\right)$-competitive algorithm for any constant $\varepsilon>0$ for 
$1|jrp, s=1, p_j=1, r_j|\sum C_j+c_\Q$.
\end{theorem}
\begin{proof}
Suppose that there is only one job arriving at $0$.
If an algorithm \new{starts} it at some time point $t$, then that algorithm cannot be better than $\frac{K+t+1}{K+1}$-competitive, because it is possible that no other jobs will arrive. 
However, if it \new{starts} the first job at $t$, then it is possible that another job  arrives at $t+1$.
In this case \new{starting the two jobs at time $t+1$ and $t+2$, respectively, with one replenishment yields a solution} of  value $K+2t+5$, while $v(S,\Q)=2K+2t+3$, thus the competitive ratio cannot be better than $\frac{2K+2t+3}{K+2t+5}$.
Observe that if $K$ is given, then the first ratio increases, while the second one decreases as $t$ increases.
This means that  we have to find a time point $\bar{t}\geq 0$ such that $\frac{K+\bar{t}+1}{K+1}=\frac{2K+2\bar{t}+3}{K+2\bar{t}+5}$, because then $\frac{K+\bar{t}+1}{K+1}$ is a lower bound on the best competitive ratio.

Some algebraic calculations show that $\bar{t}\in [K/2-5/4,K/2-1]$, thus the lower bound is at least $\frac{(3/2)K-1/4}{K+1}$. 
Therefore, for any $\varepsilon>0$, there is a sufficiently large $K$, such that there is no $(3/2-\varepsilon)$-competitive algorithm for the problem.
\end{proof}

When the job weights are arbitrary, then a slightly stronger bound can be derived.
\begin{theorem}\label{thm:onl_negative_weighted}
There is no $\left(\frac{\sqrt{5}+1}{2}-\epsilon\right)$-competitive algorithm for any constant $\epsilon > 0$ for $1|jrp, s=1, p_j=1, r_j|\sum w_jC_j+c_\Q$.
\end{theorem}

\begin{proof}
Suppose that job $j_1$ arrives at time 0 with weight $w_1=1$. If no other jobs arrive and the algorithm waits until time $t$ before \new{starting} $j_1$ then it is at least $c_1(t)$-competitive where $c_1(t)=\frac{K+t+1}{K+1}$. If another job $j_2$ arrives at time $t+1$ with weight $w_2$, then the algorithm is at least $c_2(t)$-competitive where $c_2(t)=\frac{2K+t+1+(t+2)w_2}{K+t+2+(t+3)w_2}$. To get a lower bound for an arbitrary online algorithm, we want to calculate the value of
$$
\max_{K,w_2} \min_{t}\max(c_1(t),c_2(t)).
$$	
It is easy to see that $c_1(t)$ is an increasing function of $t$, and if $K>w_2+1$ then $c_2(t)$ is a decreasing function of $t$. The second part can be proved by the following simple calculation:
$$
c_2(t+1)-c_2(t)=\frac{(w_2+1)(w_2+1-K)}{(K+t+2+(t+3)w_2)(K+t+3+(t+4)w_2)}<0
$$
if $K>w_2+1$.
So we get the best value $\overline{t} \geq 0$ when $c_1(t)=c_2(t)$.
Solving the equation, we get that
$$
\overline{t}=\frac{\sqrt{4K^2w_2-4Kw_2^2+5K^2+2Kw_2+5w_2^2+4K+4w_2}-K-3w_2-2}{2(w_2+1)}.
$$
Substituting $\overline{t}$ into $c_1(t)$ we get the following formula
\begin{multline*}
c(K,w_2)=\\ 
\frac{\sqrt{4K^2w_2-4Kw_2^2+5K^2+2Kw_2+5w_2^2+4K+4w_2}+K+2Kw_2-w_2}{2(w_2+1)(K+1)}.
\end{multline*}
\noindent
So
$$
\lim_{K\to\infty} c(K,w_2)=\frac{\sqrt{4w_2+5}+2w_2+1}{2(w_2+1)},
$$
and 
$$
\lim_{w_2\to 0} \frac{\sqrt{4w_2+5}+2w_2+1}{2(w_2+1)} =  \frac{\sqrt{5}+1}{2},
$$
which gives the desired result.
\end{proof}

\subsection{Online algorithm for $1|jrp, s=1, p_j=1, r_j|\sum F_j + c_Q$}\label{sec:online_Fj}

In this section we describe a 2-competitive algorithm for the online version of the total flow time minimization problem for the special case, where there is only one resource, and $p_j=1$ for all jobs.
Observe that this problem differs from the problem of the previous section only in the objective function.
Hence, we can also suppose that the order of the jobs is the same in any schedule, and $K:=K_0+K_1$.

The algorithm is almost the same as in the previous section and several parts of the proof are analogous.
\new{We also use the  notations of the previous section, e.g., $\js_{t}, G(a)$, etc.
Observe that if we start the jobs of $\js_{t}$ from $t$, then their total flow time will be $\sum_{j\in\js_{t}}(t-r_j)+G(|\js_{t}|)$.
}

\medskip
\begin{varalgorithm}{Online{\_}SumFj}
\begin{small}
\caption{}
\label{alg:online_flow}

Initialization: $t:=0$
\begin{enumerate}
\item \new{Determine the set $\js_{t}$ of unscheduled jobs at time $t$.}\label{step:onlinesumflow_first}
\item \new{If $\sum_{j\in\js_{t}}(t-r_j)+G(|\js_{t}|)\geq K$, then replenish the resource, start the jobs of $\js_{t}$ from $t$ and let $t:=t+|\js_{t}|$. Otherwise, $t:=t+1$.}\label{step:onlineflow_sumwork}
\item Go to step~\ref{step:onlinesumflow_first}. \label{step:onlwinesumflow_last}
\end{enumerate}
\end{small}
\end{varalgorithm}


Let $(S,\Q)$ denote the solution created by Algorithm~\ref{alg:online_flow}, while $(S^\star,\Q^\star)$ an arbitrary optimal solution.
The notations $t_i, \new{\js_{(i)},b}_i,y_i,z_i$ have the same meaning as in the previous section,  see again  Figure~\ref{fig:online_S} for illustration.
The next observation is analogous to Observation~\ref{obs:onl_schrule}:

\begin{observation}\label{obs:onl_schrule_flow}
If the machine is idle in $[t_i-1,t_i]$, then $\sum_{j\in \new{\js_{t_i-1}}}(t_i-1-r_j)+\new{G(y_i)}<K$.
\end{observation}

Let $ALG_i:=K+\sum_{j\in\new{\js_{(i)}}}F_j$ ($i\geq 1$), thus we have $v(S,\Q)=\sum_{i\geq 1}ALG_i$.
Let $s_i$ denote the number of replenishments in $\Q^\star$ in $[t_{i-1}+1,t_i-1]$, and let $OPT_i:=s_iK+\sum_{j\in\new{\js_{(i)}}}F^\star_j$.
Observe that $v(S^\star,\Q^\star)\geq\sum_{i\geq 1}OPT_i$. We are ready to prove the main result of this section.

\begin{theorem}\label{thm:onl_flow}
Algorithm~\ref{alg:online_flow} is 2-competitive for the online \linebreak $1|jrp, s=1, p_j=1, r_j|\sum F_j+c_\Q$ problem.
\end{theorem}

\begin{proof}Analogously to the proof of Theorem~\ref{thm:onl_2}, we prove $\sum_{\mu=0}^{\ell} ALG_{i+\mu}\leq 2\cdot \sum_{\mu=0}^{\ell} OPT_{i+\mu}$ for each $i$ such that in the schedule constructed by the algorithm,
the machine is idle in $[t_i-1,t_i)$, and the next idle period starts in $[t_{i+\ell},t_{i+\ell+1})$, from which the theorem follows.

Recall the definition of $\ell'$ from the proof of Theorem~\ref{thm:onl_2}. Then, the next claim is analogous to Claim~\ref{clm:onl_2comp}.

\begin{claim}\label{clm:onl_2flow}
$ALG_{i+\mu}\leq 2\cdot OPT_{i+\mu}$ for all $\mu\in[\ell',\ell]$.
\end{claim}
\begin{proof}
If $\ell'=\ell+1$, then the claim is trivial.

Otherwise, we have $r_j\geq t_{i+\ell'-1}+1=:t'$ for each job $j\in\cup_{\mu=\ell'}^{\ell}\new{\js_{(i+\mu)}}$, hence they cannot be \new{started} before the first replenishment after $t'$.
The first replenishment after $t'$ in $\Q^\star$ is not earlier than $t_{i+\ell'}$ in $S^\star$ due to the definition of $\ell'$.
We have assumed that the order of the jobs is the same in every schedule, and there is no idle time among the jobs in $\cup_{\mu=\ell'}^{\ell}\new{\js_{(i+\mu)}}$ in $S$, thus, we have $C^\star_j\geq C_j$ and $F^\star_j\geq F_j$ for each $j\in \cup_{\mu=\ell'}^{\ell}\new{\js_{(i+\mu)}}$.
Since $OPT_{i+\mu}\geq\sum_{j\in\new{\js_{(i+\mu)}}}F^\star_j\geq \sum_{j\in\new{\js_{(i+\mu)}}}F_j\geq K$ for all $\mu\in[\ell',\ell]$, where the last inequality follows from the condition of step~\ref{step:onlineflow_sumwork} of Algorithm~\ref{alg:online_flow}, we have $ALG_{i+\mu}=K+\sum_{j\in\new{\js_{(i+\mu)}}}F_j\leq 2\cdot OPT_{i+\mu}$ for all $\mu\in[\ell',\ell]$.
\end{proof}

If $\ell'=0$, then \new{our work is done}, thus from now on we suppose $\ell'\geq 1$.
It \new{remains} to prove  $\sum_{\mu=0}^{\ell'-1} ALG_{i+\mu}\leq 2\cdot \sum_{\mu=0}^{\ell'-1} OPT_{i+\mu}$, if there is at least one replenishment in every interval $[t_{i+\mu-1}+1,t_{i+\mu}-1]$, $\mu=0,\ldots,\ell'-1$ in $\Q^\star$, i.e., $s_{i+\mu}\geq 1$ for every $\mu=0,\ldots,\ell'-1$.

First, we prove $ALG_i\leq 2\cdot OPT_i$.
Let $\tau_{i,0}:=t_{i-1}$, and  $\tau_{i,1}\new{<}\tau_{i,2}\new{<}\ldots\new{<} \tau_{i,s_i}$ be the time points of the replenishments in $[t_{i-1}+1,t_i-1]$ in $\Q^\star$.
Let $\tau_{i,s_i+1}:=t_i$.
For $k=1,2,\ldots,s_i+1$, let $\new{\js_{(i)}}(k):=\{j\in \new{\js_{(i)}}: r_j\in(\tau_{i,k-1},\min\{t_i-1,\tau_{i,k}\}]\}\subseteq \new{\js_{t_i-1}}$ \new{and $b_i(k):=|\js_{(i)}(k)|$.}
See Figure~\ref{fig:flowtime_optrepl} for illustration.
Note that these sets are pairwise disjoint and  $\bigcup_{k= 1}^{s_i+1}\new{\js_{(i)}}(k)=\{j\in\new{\js_{(i)}}:r_j<t_i\} = \new{\js_{t_i-1}}$, and their \new{total} size is $y_i$.

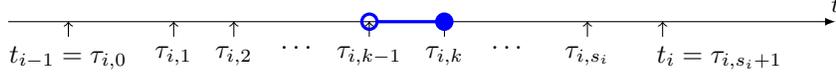
\begin{figure}
\begin{tikzpicture}
\def\ox{0} 
\def\oy{0} 
\def\ui{1}
\def\uii{6.50}
\def\uiii{8}
\coordinate(o) at (\ox,\oy); 
\coordinate(u1) at (\ui,\oy);
\coordinate(u2) at (\uii,\oy);
\coordinate(u3) at (\uiii,\oy);

\tikzstyle{mystyle}=[draw, minimum height=0.5cm,rectangle, inner sep=0pt,font=\scriptsize]

\def\tl{11.0} 
\def\oyi{0}
\draw [-latex](\ox,\oyi) node[above left]{} -- (\ox+\tl,\oyi) node[above,font=\small]{$t$};

\coordinate (uq1) at (\ui,\oyi);
\coordinate (uq2) at (\uii,\oyi);
\coordinate (uq3) at (\uiii,\oyi);
\draw[<-] ($(uq1)-(0.2,0.0)$) -- ($(uq1)-(0.2,0.2)$) node[below] {$t_{i-1}=\tau_{i,0}$};
\draw[<-] ($(uq1)-(-1.2,0.0)$) -- ($(uq1)-(-1.2,0.2)$) node[below] {$\tau_{i,1}$};
\draw[<-] ($(uq1)-(-2.0,0.0)$) -- ($(uq1)-(-2.0,0.2)$) node[below] {$\tau_{i,2}$} node[below right=0cm and 0.5cm]{$\ldots$};
\draw[<-] ($(uq2)-(1.7,0.0)$) -- ($(uq2)-(1.7,0.2)$) node[below] {$\tau_{i,k-1}$};
\draw[<-] ($(uq2)-(0.7,0.0)$) -- ($(uq2)-(0.7,0.2)$) node[below] {$\tau_{i,k}$} node[below right=0cm and 0.5cm]{$\ldots$};
\draw[<-] ($(uq2)+(1.2,0.0)$) -- ($(uq2)+(1.2,-0.2)$) node[below] {$\tau_{i,s_i}$};
\draw[<-] ($(uq2)+(2.2,0.0)$) -- ($(uq2)+(2.2,-0.2)$) node[below right=0cm and -0.2cm] {$t_{i}=\tau_{i,s_i+1}$};

\draw[very thick, blue,o-*] ($(uq2)-(1.82,0.0)$)  -- ($(uq2)-(0.58,0.0)$);

\end{tikzpicture}
\caption{Replenishment times in $[t_{i-1}+1,t_i-1]$ in $\Q^\star$. The release dates of the jobs of $\new{\js_{(i)}}(k)$ are from the blue interval.}\label{fig:flowtime_optrepl}
\end{figure}

Each job in  $\new{\js_{(i)}}(k)$ gets resource at $\tau_{i,k}$ in $\Q^\star$, if $k\leq s_i$, thus they cannot start earlier  than $\tau_{i,k}$ in $S^\star$. 
The earliest start time of a job from $\new{\js_{(i)}}(s_i+1)$ is $t_i$.
Since $F_j=(\tau_{i,k}-r_j)+(C_j-\tau_{i,k})$, thus we have $\sum_{j\in\new{\js_{(i)}}(k)}F_j\geq \sum_{j\in\new{\js_{(i)}}(k)}(\tau_{i,k}-r_j)+ \new{G(b_{(i)}(k))}$  for any $1\leq k\leq s_{i}+1$.
Applying the previous inequality for each $k$ and $\sum_{j\in\new{\js_{(i)}}\setminus\new{\js_{t_i-1}}}F_j\geq \new{G(z_i)}$, we get
\begin{align*}
OPT_i&\geq s_iK + \sum_{k= 1}^{s_i+1}\sum_{j\in\new{\js_{(i)}}(k)}(\tau_{i,k}-r_j)+\sum_{k= 1}^{s_i+1} \new{G(\new{b_{(i)}}(k))} +\new{G(z_i)}\\
&\geq s_iK + \sum_{k= 1}^{s_i+1}\sum_{j\in\new{\js_{(i)}}(k)}(\tau_{i,k}-r_j) + \new{G\left(\frac{y_i}{s_i+1}\right)}\cdot (s_i+1)+\new{G(z_i)},
\end{align*}
where the last inequality follows from simple algebraic rules.
Furthermore, we have
\begin{align*}
&ALG_i=K+\sum_{j\in\new{\js_{(i)}}}F_j=K+\sum_{j\in\new{\js_{t_i-1}}}(t_i-r_j)+\new{G(y_i)}+y_iz_i+\new{G(z_i)},
\end{align*}
thus
\begin{align*}
&2\cdot OPT_i-ALG_i\geq (2s_i-1)K+2\cdot\sum_{k= 1}^{s_i+1}\sum_{j\in\new{\js_{(i)}}(k)}(\tau_{i,k}-r_j) + \\ 
&+ y_i\left(\frac{y_i}{s_i+1}+1\right)+
\new{G(z_i)}-\sum_{j\in\new{\js_{t_i-1}}}(t_i-r_j)-\new{G(y_i)}-y_iz_i 
\end{align*}
\begin{align*}
& \geq(2s_i-2)K+2\cdot\sum_{k= 1}^{s_i+1}\sum_{j\in\new{\js_{(i)}}(k)}(\tau_{i,k}-r_j)+ y_i\left(\frac{y_i}{s_i+1}+1\right)+\new{G(z_i)}-\\
&-y_i-y_iz_i=(2s_i-2)K+2\cdot\sum_{k= 1}^{s_i+1}\sum_{j\in\new{\js_{(i)}}(k)}(\tau_{i,k}-r_j)+\frac{y_i^2}{s_i+1}+\new{G(z_i)}-y_iz_i,
\end{align*}
where the second inequality follows from Observation~\ref{obs:onl_schrule_flow} and from $|\new{\js_{t_i-1}}|=y_i$.

On the one hand, if $s_i=1$, then $2\cdot OPT_i-ALG_i\geq y_i^2/2+z_i^2/2-y_iz_i\geq 0$.
On the other hand, if $s_i\geq 2$, then we use again Observation~\ref{obs:onl_schrule_flow}, and we get $2\cdot OPT_i-ALG_i\geq (2s_i-3)K+\new{G(y_i)}+\new{G(z_i)}-y_iz_i\geq y_i^2/2+z_i^2/2-y_iz_i\geq 0$.
Therefore, we have proved $ALG_i\leq 2\cdot OPT_i$.

Now we prove $ALG_{i+\mu}\leq 2\cdot OPT_{i+\mu}$ for any $1\leq \mu\leq \ell'-1$.
Let $1\leq \mu\leq \ell'-1$,  and $j'\in\new{\js_{(i+\mu)}}$ be arbitrary.
Suppose that  $j'$ has the $n_{j'}^{th}$ smallest release date among the jobs in $\new{\js_{(i+\mu)}}$, i.e., it is the $\left(\sum_{\nu=1}^{i+\mu-1}\new{b_{(\nu)}}+n_{j'}\right)^{th}$ job  in the fixed non-decreasing release date order that determine the order of the jobs in any schedule.
Observe that $C_{j'}=t_i+\sum_{\nu=0}^{\mu-1}\new{b_{(i+\nu)}}+n_{j'}$, because the algorithm \new{starts} all of the jobs of $\bigcup_{\nu=0}^{\mu-1}\new{\js_{(i+\nu)}}$ from $t_i$ without any gap, and after that, it starts the jobs from $\new{\js_{(i+\mu)}}$ also without any gap.
However, it is possible that $y_i$ jobs from $\new{\js_{(i)}}$ are \new{started} before $t_i$ in $S^\star$, but every other job from  $\bigcup_{\nu=0}^{\mu}\new{\js_{(i+\nu)}}$ has a release date at least $t_i$.
Hence, we have $C^\star_{j'}\geq t_i+z_i+\sum_{\nu=1}^{\mu-1}\new{b_{(i+\nu)}}+n_{j'}$, since the order of the jobs is fixed.
Since $\new{b_{(i)}}=y_i+z_i$, we have $C_{j'}\leq C^\star_{j'}+y_i$ and $F_{j'}\leq F^\star_{j'}+y_i$. 

Let $h:=\new{b_{(i+\mu)}}$.
Since $\sum_{j\in\new{\js_{(i+\mu)}}}F^\star_j\geq \new{G(h)}$, $OPT_{i+\mu}\geq K+\sum_{j\in\new{\js_{(i+\mu)}}}F^\star_j$, and  $ALG_i\leq K+\sum_{j\in\new{\js_{(i+\mu)}}}(F^\star_j+y_i)$, we have 
\begin{align*}
&2\cdot OPT_{i+\mu}-ALG_{i+\mu}\geq K+\new{G(h)}-hy_i\geq \new{G(y_i)}+\new{G(h)}-hy_i\geq 0,
\end{align*}
where the second inequality follows from Observation~\ref{obs:onl_schrule_flow}.

Finally, observe that the problem instance used at the end of Theorem~\ref{thm:onl_2} (there is only one job $j$, with $r_j=0$) shows that the above analysis is tight, i.e., Algorithm~\ref{alg:online_flow} is not $\alpha$-competitive for any $\alpha<2$.

\end{proof}

The next theorem gives a lower bound on the best possible competitive ratio.

\begin{theorem}\label{thm:online_sum_fj_3/2}
There does not exist an online algorithm for $1|jrp, s=1,p_j=1, r_j|\sum F_j+c_\Q$ with competitive ratio better than $\frac{3}{2} - \epsilon$ for any constant $\epsilon > 0$.
\end{theorem}

\begin{proof}
Suppose that  one job is released at $0$, and some online algorithm $ALG$ \new{starts} it at time $t\geq 0$, then another job at time $t+1$ is released. The offline optimum for this problem instance with 2 jobs is $OPT = \min(2K+2, K+t+2)$, while the online solution has a cost of $2K+t+2$.

\begin{enumerate}
\item If $K < t$, then  $OPT = 2K+2$, therefore $\frac{ALG}{OPT} \geq \frac{3}{2}$.
\item If $K \geq t+2$, then $OPT = K+t+2, \frac{ALG}{OPT} = \frac{2K+t+2}{K+t+2} \geq \frac{3}{2}$.
\item If $t \leq K \leq t+1$, then for any $\varepsilon> 0$ we have $\frac{ALG}{OPT} \geq \frac{3K+2}{2K+2} > \frac32-\varepsilon$  if $K$ is sufficiently large.
\end{enumerate}
\end{proof}

\section{Competitive analysis for the online $1|jrp, r_j|\Fmax + c_Q$ problem}
\label{sec:online_fmax}
Throughout this section we assume $p_j = 1$ for each job $j$, and there is a single resource only ($s = 1$). We will describe a $\sqrt{2}$-competitive algorithm for a  semi-online variant, which we call \textit{\new{regular} input},
where $r_j = j$ for $j \geq 1$. That is, the job release dates are known in advance, but we do not know how many jobs will arrive.
We also provide lower bounds for the competitive ratio of any algorithm for \new{regular} as well as general input. 

For the competitive analysis we need a lower bound for the offline optimum, which is the topic of the next section.

\subsection{Lower bounds for the offline optimum}

Let  $\psum = \sum_{j=1}^n {p_j}$.

\begin{lemma}\label{lem:lower_bound_min}
$\min_F (K \big \lceil \frac{p_{\mathrm{sum}}}{F} \big \rceil + F)$ is a lower bound for the offline optimum.
\end{lemma}

\begin{proof}
If in a solution the maximum flow time of jobs is $F$, then at least $\big \lceil \frac{\psum}{F} \big \rceil$ replenishment is needed, which gives the lower bound above.
\end{proof}

\begin{lemma}\label{lem:lower_bound_sqrt}
$2 \sqrt{K \psum}$ is a lower bound for the offline optimum.
\end{lemma}

\begin{proof}
Consider the formula from Lemma~\ref{lem:lower_bound_min} without the ceiling function, which is also a lower bound for the optimum: $\min_F (K  \frac{\psum}{F} + F)$. The minimum is obtained in $F = \sqrt{Kp_{sum}}$, and it has a value of $2 \sqrt{K \psum}$.
\end{proof}

If $p_j = 1$ for every job, \new{we get} $\min_F (K \lceil \frac{n}{F} \rceil + F)$ and $2 \sqrt{K n}$ as lower bounds. 

Note that for a \new{regular} input, the first lower bound will give exactly the offline optimum. However, for a general input, the above lower bounds can be weak.

\subsection{Online algorithm for \new{regular} input}

For \new{regular} input, we can give the exact formula for the offline optimum using the lower bound obtained in Lemma~\ref{lem:lower_bound_sqrt}. However, in the online problem we don't know the total number of jobs in the input, therefore we propose an algorithm that determines the replenishment times in advance.

We  distinguish the cases $K = 1$ and $K > 1$. Let
\[t_i :=
  \begin{cases}
    \frac{i(i+1)}{2} &  \text{if } K=1, \\
    K \frac{i^2 + 3i}{2} &  \text{if } K>1.
  \end{cases}
\]

\begin{varalgorithm}{Max{\_}Flowtime}
\new{
\begin{small}
\caption{}\label{alg:online_fmax}

Initialization: $t:=1$ and $i:=1$.
\begin{enumerate}
\item If $t = t_i$ or no job arrives at $t$, then \label{step:online_fmax_first}
\begin{enumerate}[a.]
\item Replenish the resource, and start the yet unscheduled jobs from $t$.
\item If there are no more jobs, then stop; otherwise
let $t:=t+(t_i-t_{i-1})$, $i:=i+1$.
\end{enumerate}
Else  $t:=t+1$.
\item Go to step~\ref{step:online_fmax_first}. \label{step:online_fmax_last}
\end{enumerate}
\end{small}
}
\end{varalgorithm}

\new{In Algorithm~\ref{alg:online_fmax}, the replenishments occur at the time points $t_i$, and also after the release date of the last job. 
At each replenishment all yet unscheduled jobs are put on the machine in increasing release date order.} The schedule obtained for the first few jobs when $K=1$ is shown in Figure~\ref{fig:online_fmax}.

\begin{figure}[!th]
\begin{tikzpicture}
\def\ox{0} 
\def\oy{0} 
\def\ui{0}
\def\uii{2}
\def\uiii{5}
\def\uiiii{9}
\coordinate(o) at (\ox,\oy); 
\coordinate(u1) at (\ui,\oy);
\coordinate(u2) at (\uii,\oy);
\coordinate(u3) at (\uiii,\oy);
\coordinate(u4) at (\uiiii,\oy);

\tikzstyle{mystyle}=[draw, minimum height=0.5cm,rectangle, inner sep=0pt,font=\scriptsize]

\def\tl{11} 
\def\oyi{0}
\draw [-latex](\ox,\oyi) node[above left]{$S$} -- (\ox+\tl,\oyi) node[above,font=\small]{$t$};

\coordinate (uq1) at (\ui,\oyi);
\coordinate (uq2) at (\uii,\oyi);
\coordinate (uq3) at (\uiii,\oyi);
\coordinate (uq4) at (\uiiii,\oyi);
\draw[<-] (uq1) -- ($(uq1)-(0,0.2)$) node[below] {$t_{1} = 1$};
\draw[<-] (uq2) -- ($(uq2)-(0,0.2)$) node[below] {$t_2=3$};
\draw[<-] (uq3) -- ($(uq3)-(0,0.2)$) node[below] {$t_3=6$};
\draw[<-] (uq4) -- ($(uq4)-(0,0.2)$) node[below] {$t_4=10$};

\node(b1) [above right=0.0cm and 0.0cm of u1,mystyle, minimum width=0.9 cm]{$j_1$};
\node(b2) [above right=0.0cm and 0.0cm of u2,mystyle, minimum width=0.9 cm]{$j_2$};
\node(b22) [right=0.0cm of b2,mystyle, minimum width=0.9 cm]{$j_3$};
\node(b3) [above right=0.0cm and 0.0cm of u3,mystyle, minimum width=0.9 cm]{$j_4$};
\node(b32) [right=0.0cm of b3,mystyle, minimum width=0.9 cm]{$j_5$};
\node(b33) [right=0.0cm of b32,mystyle, minimum width=0.9 cm]{$j_6$};
\node(b4) [above right=0.0cm and 0.0cm of u4,mystyle, minimum width=0.9 cm]{$j_7$};
\node(b42) [right=0.0cm of b4,mystyle, minimum width=0.9 cm]{$j_8$};

\end{tikzpicture}
\caption{Schedule $S$ created by Algorithm~\ref{alg:online_fmax} for $K=1$.}\label{fig:online_fmax}
\end{figure}
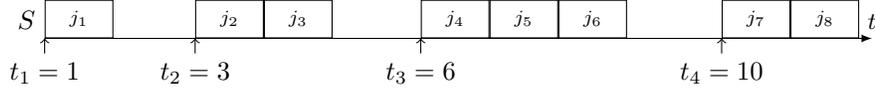

\begin{theorem}\label{thm:online_fmax_sqrt2}
Algorithm~\ref{alg:online_fmax} is $\sqrt{2}$-competitive on \new{regular} input.
\end{theorem}

\begin{proof}
We will prove the theorem separately for $K=1$ and $K>1$.

For $K=1$, first we check the competitive ratio for \new{the case $n=t_i$}. The maximum flow time is given by the  job \new{started} in $t_i$, therefore the total cost of the algorithm is: $i + t_i - t_{i-1} = 2 i$.

Using the lower bound for the offline optimum obtained in Lemma~\ref{lem:lower_bound_sqrt}, it is enough to check if
$2\sqrt{2} \sqrt{\frac{i(i+1)}{2}} \geq 2i$,
which holds for each $i \geq 1$.

Now suppose we have \new{$n=t_i + j$} jobs, where $1 \leq j < t_{i+1}- t_i = i+1$. In the online algorithm the last replenishment is at \new{$t_i + j$}, and the cost is: $i + 1 + \max\{t_i - t_{i-1},j\} = 2i + 1$. It is easy to verify that 
$
 2\sqrt{2} \sqrt{\frac{i(i+1)}{2} + j} \geq 2i+1,
$
\new{hence}, the algorithm is $\sqrt{2}$-competitive for $K=1$.

For $K > 1$, there are three cases to consider : $n < 2K$, $n = t_i$ and $n = t_i + j$ for $1\leq j < t_{i+1} - t_i = K(i+2)$.

If $n < 2K$, \new{we obtain that $ALG = K+n = OPT$}.

If $n = t_i$, the total cost of the algorithm is $Ki + t_i - t_{i-1} = K(2i+1)$.
Thus, we have to compare it to the lower bound from Lemma~\ref{lem:lower_bound_sqrt}, which is $2K \sqrt{\frac{i^2 + 3i}{2}}$.
It is easy to verify that $2\sqrt{2}K \sqrt{\frac{i^2 + 3i}{2}} \geq K(2i+1)$ for every $i,K \geq 1$.

If $n = t_i+j$ for $1 \leq j < (t_{i+1} - t_i) = K(i+2)$, the total cost of the algorithm is
 $K(i+1) + \max\{j,t_i - t_{i-1}\} = K(i+1) + \max\{j,K(i+1)\}$. 
 
For $1 \leq j \leq K(i+1)$, one can easily verify that $2\sqrt{2}K \sqrt{\frac{i^2 + 3i}{2} + j} \geq K(2i+2)$, and for $K(i+1) < j < K(i+2)$,  $2\sqrt{2}K \sqrt{\frac{i^2 + 3i}{2} + j} \geq K(2i+3)$ hold, from which $\sqrt{2} OPT \geq ALG$ follows. 

Therefore, we have shown that the algorithm is $\sqrt{2}$-competitive. \new{The analysis is tight}, because one can upper bound the offline optimum for $n$ jobs with $2 \sqrt{Kn} + K$ for any fixed $K$, and for $n = t_i$ we have:

\begin{enumerate}
\item If $K = 1$, then $ \frac{ALG}{OPT} \geq \frac{2i}{\sqrt{2i(i+1)} + 1} \rightarrow \sqrt{2} \text{ if } i \rightarrow \infty.$
\item If $K > 1$, then $ \frac{ALG}{OPT} \geq \frac{K(2i+1)}{K\sqrt{2(i^2+3i)} + K} \rightarrow \sqrt{2} \text{ if } i \rightarrow \infty.$
\end{enumerate}
\vskip -12pt
\end{proof}

Algorithm \ref{alg:online_fmax} belongs to a family of online algorithms, where the first replenishment occurs at $t_1$ (assuming $n \geq t_1$), and after scheduling the jobs ready to be started, it leaves a gap of length $K \delta$, where $\delta$ is some fixed positive integer, before replenishing again (unless there are no more jobs, in which case the algorithm replenishes immediately, and start the remaining jobs). 

In Algorithm~\ref{alg:online_fmax}  we have $t_1 = 1$ for $K=1$ and $t_1 = 2K$ for $K>1$, and $\delta = 1$ for both cases. By simple calculations it can also be shown that these type of algorithms can have the best competitive ratio of $\sqrt{2}$, therefore our algorithm is the best possible among them.

\begin{theorem}\label{thm:online_fmax_4/3}
On \new{regular} input, there does not exist an online algorithm with competitive ratio better than $\frac{4}{3}$.
\end{theorem}

\begin{proof}
\new{Suppose an online algorithm makes the first replenishment at time point $t$, and consider a regular input with $t+1$ jobs}. Then $ALG = 2K + t$, since there is one replenishment in $t$ serving $t$ jobs with maximum flow time of $t$, and there is another replenishment in $t+1$ serving only one job. We distinguish three cases for the possible values of $t$:

\begin{enumerate}
\item If $t \leq 2K-4$, then $OPT \leq K+t+1$, \new{since this is the cost of a solution with one replenishment in $t+1$}, therefore: 
\[ \frac{ALG}{OPT} \geq \frac{2K+t}{K+t+1} \geq \frac{4}{3}.\]
\item If $2K-3 \leq t \leq 2K+3$, then $ALG = 4K+i$, where $i = t-2K \in [-3,3]$, and \new{$OPT \leq 3K+3$, since this is the cost of a solution with one replenishment}. Therefore,  
\[
 \frac{ALG}{OPT} \rightarrow \frac{4}{3}, \text{ if } K \rightarrow \infty.
 \]
\item If $t \geq 2K+4$, then  \new{ $OPT \leq 2K + \lceil \frac{t+1}{2} \rceil$, since this is the cost of a solution with two replenishments}, therefore: 
\[\frac{ALG}{OPT} \geq \frac{2K+t}{2K + \lceil \frac{t+1}{2} \rceil} \geq \frac{2K+t}{2K + \frac{t+1}{2} + 1} \geq \frac{4}{3}.\]
\end{enumerate}
\vskip -12pt
\end{proof}

\subsection{Lower bound on the competitive ratio}
If the input is not \new{regular} (i.e., the release dates of the jobs are arbitrary), then we have a higher lower bound:

\begin{theorem}\label{thm:onl_fmax_negative_general_case}
There is no $\left(\frac{\sqrt{5}+1}{2}-\epsilon\right)$-competitive online algorithm for any constant $\epsilon > 0$ for $1|jrp, s=1, p_j=1, r_j|\Fmax+c_\Q$.
\end{theorem}

\begin{proof}
Consider an arbitrary algorithm.
Suppose that job $j_1$ arrives at time 0.
If no more jobs arrive, and the algorithm \new{starts} $j_1$ at $t$, then it is at least $c_1(t,K)$-competitive, where $c_1(t,K)=\frac{K+t+1}{K+1}$.
If $t$ further jobs arrive at  $t+1$, then the algorithm is at least $c_2(t,K)$-competitive, where $c_2(t,K)=\frac{2K+t+1}{K+t+2}$. 
To get a lower bound for an arbitrary online algorithm, we want to calculate the value of
$$
\max_{K} \min_{t}\max(c_1(t,K),c_2(t,K)).
$$	
It is easy to see that $c_1(t,K)$ is an increasing function of $t$, while $c_2(t,K)$ is a decreasing function of $t$. 
Therefore, $\max(c_1(t,K),c_2(t,K))$  is minimal for a $\bar{t}\geq 0$ if  $c_1(\bar {t},K)=c_2(\bar{t},K)$.
Some algebraic calculations show that this happens if $\bar{t}=\frac{\sqrt{5K^2+4K}-K-2}{2}$.
Substituting $\bar{t}$ into $c_1(t,K)$, we get $\lim_{K\to\infty} c_1(\bar{t},K)=\frac{\sqrt{5}+1}{2}$.
 Therefore, for any $\varepsilon>0$, there is a sufficiently large $K$, such that there is no $\left(\frac{\sqrt{5}+1}{2}-\epsilon\right)$-competitive algorithm for the problem.
\end{proof}

\section{Conclusions}
\label{sec:conclude}
In this paper we have combined single machine scheduling with the joint replenishment problem, where the processing of any job on the machine can only be started if the required item types are ordered after the release date of the job.
We have proved complexity results, and devised polynomial time algorithms both for the offline and the online variant of the problem.
However, several open questions remained, and we list some of the most intriguing ones. Is the offline problem with the $\sum w_j C_j$ objective solvable in polynomial time when $p_j=p$ and $s$ is constant, while the $w_j$ are arbitrary? What is the best competitive ratio for the considered  online problems? \new{In particular, for the online problem with regular input and maximum flow time objective, only the number of the jobs is unknown, yet, there is a gap between the best upper and lower bound.} What can we say in case of more complex machine environments?

\section*{Acknowledgment}

The authors are grateful to the anonymous reviewers for constructive comments that helped to improve the presentation.

\section*{Conflict of interest}

The authors declare that they have no conflict of interest.

\bibliographystyle{spmpsci}      
\bibliography{joint_rep}   

\begin{thebibliography}{10}
\providecommand{\url}[1]{{#1}}
\providecommand{\urlprefix}{URL }
\expandafter\ifx\csname urlstyle\endcsname\relax
  \providecommand{\doi}[1]{DOI~\discretionary{}{}{}#1}\else
  \providecommand{\doi}{DOI~\discretionary{}{}{}\begingroup
  \urlstyle{rm}\Url}\fi

\bibitem{Afrati1999}
Afrati, F., Bampis, E., Chekuri, C., Karger, D., Kenyon, C., Khanna, S., Milis,
  I., Queyranne, M., Skutella, M., Stein, C., Sviridenko, M.: Approximation
  schemes for minimizing average weighted completion time with release dates.
\newblock In: 40th Annual Symposium on Foundations of Computer Science (Cat.
  No. 99CB37039), pp. 32--43. IEEE (1999)

\bibitem{anderson2004}
Anderson, E.J., Potts, C.N.: Online scheduling of a single machine to minimize
  total weighted completion time.
\newblock Mathematics of Operations Research \textbf{29}(3), 686--697 (2004)

\bibitem{arkin1989}
Arkin, E., Joneja, D., Roundy, R.: Computational complexity of uncapacitated
  multi-echelon production planning problems.
\newblock Operations research letters \textbf{8}(2), 61--66 (1989)

\bibitem{averbakh2013approximation}
Averbakh, I., Baysan, M.: Approximation algorithm for the on-line
  multi-customer two-level supply chain scheduling problem.
\newblock Operations Research Letters \textbf{41}(6), 710--714 (2013)

\bibitem{averbakh2007line}
Averbakh, I., Xue, Z.: On-line supply chain scheduling problems with
  preemption.
\newblock European Journal of Operational Research \textbf{181}(1), 500--504
  (2007)

\bibitem{azar2016make}
Azar, Y., Epstein, A., Je{\.z}, L., Vardi, A.: Make-to-order integrated
  scheduling and distribution.
\newblock In: Proceedings of the Twenty-Seventh Annual ACM-SIAM Symposium on
  Discrete Algorithms, pp. 140--154. SIAM (2016)

\bibitem{bansal2007}
Bansal, N., Dhamdhere, K.: Minimizing weighted flow time.
\newblock ACM Transactions on Algorithms (TALG) \textbf{3}(4), 39--es (2007)

\bibitem{baptiste2000}
Baptiste, P.: Scheduling equal-length jobs on identical parallel machines.
\newblock Discrete Applied Mathematics \textbf{103}(1-3), 21--32 (2000)

\bibitem{becchetti2009latency}
Becchetti, L., Marchetti-Spaccamela, A., Vitaletti, A., Korteweg, P., Skutella,
  M., Stougie, L.: Latency-constrained aggregation in sensor networks.
\newblock ACM Transactions on Algorithms (TALG) \textbf{6}(1), 1--20 (2009)

\bibitem{bienkowski2015}
Bienkowski, M., Byrka, J., Chrobak, M., Dobbs, N., Nowicki, T., Sviridenko, M.,
  {\'S}wirszcz, G., Young, N.E.: Approximation algorithms for the joint
  replenishment problem with deadlines.
\newblock Journal of Scheduling \textbf{18}(6), 545--560 (2015)

\bibitem{bienkowski2014}
Bienkowski, M., Byrka, J., Chrobak, M., Je{\.z}, {\L}., Nogneng, D., Sgall, J.:
  Better approximation bounds for the joint replenishment problem.
\newblock In: Proceedings of the twenty-fifth annual ACM-SIAM symposium on
  Discrete algorithms, pp. 42--54. SIAM (2014)

\bibitem{bosman2020}
Bosman, T., Olver, N.: Improved approximation algorithms for inventory
  problems.
\newblock In: International Conference on Integer Programming and Combinatorial
  Optimization, pp. 91--103. Springer (2020)

\bibitem{buchbinder13}
Buchbinder, N., Kimbrel, T., Levi, R., Makarychev, K., Sviridenko, M.: Online
  make-to-order joint replenishment model: Primal-dual competitive algorithms.
\newblock Operations Research \textbf{61}(4), 1014--1029 (2013)

\bibitem{chekuri2001b}
Chekuri, C., Khanna, S., Zhu, A.: Algorithms for minimizing weighted flow time.
\newblock In: Proceedings of the thirty-third annual ACM symposium on Theory of
  computing, pp. 84--93 (2001)

\bibitem{chekuri2001}
Chekuri, C., Motwani, R., Natarajan, B., Stein, C.: Approximation techniques
  for average completion time scheduling.
\newblock SIAM Journal on Computing \textbf{31}(1), 146--166 (2001)

\bibitem{chen2010integrated}
Chen, Z.L.: Integrated production and outbound distribution scheduling: review
  and extensions.
\newblock Operations research \textbf{58}(1), 130--148 (2010)

\bibitem{cheung2016}
Cheung, M., Elmachtoub, A.N., Levi, R., Shmoys, D.B.: The submodular joint
  replenishment problem.
\newblock Mathematical Programming \textbf{158}(1-2), 207--233 (2016)

\bibitem{epstein2001}
Epstein, L., Van~Stee, R.: Lower bounds for on-line single-machine scheduling.
\newblock In: International Symposium on Mathematical Foundations of Computer
  Science, pp. 338--350. Springer (2001)

\bibitem{goemans2002}
Goemans, M.X., Queyranne, M., Schulz, A.S., Skutella, M., Wang, Y.: Single
  machine scheduling with release dates.
\newblock SIAM Journal on Discrete Mathematics \textbf{15}(2), 165--192 (2002)

\bibitem{graham1979}
Graham, R.L., Lawler, E.L., Lenstra, J.K., Rinnooy~Kan, A.: Optimization and
  approximation in deterministic sequencing and scheduling: a survey.
\newblock Annals of Discrete Mathematics \textbf{5}, 287--326 (1979).
\newblock \doi{10.1016/S0167-5060(08)70356-X}

\bibitem{hoogeveen1996}
Hoogeveen, J.A., Vestjens, A.P.: Optimal online algorithms for single-machine
  scheduling.
\newblock In: International Conference on Integer Programming and Combinatorial
  Optimization, pp. 404--414. Springer (1996)

\bibitem{karlin1994competitive}
Karlin, A.R., Manasse, M.S., McGeoch, L.A., Owicki, S.: Competitive randomized
  algorithms for nonuniform problems.
\newblock Algorithmica \textbf{11}(6), 542--571 (1994)

\bibitem{karlin1988competitive}
Karlin, A.R., Manasse, M.S., Rudolph, L., Sleator, D.D.: Competitive snoopy
  caching.
\newblock Algorithmica \textbf{3}(1), 79--119 (1988)

\bibitem{karp1992line}
Karp, R.M.: On-line algorithms versus off-line algorithms: How much is it worth
  to know the future?
\newblock In: IFIP congress (1), vol.~12, pp. 416--429 (1992)

\bibitem{kellerer1999}
Kellerer, H., Tautenhahn, T., Woeginger, G.: Approximability and
  nonapproximability results for minimizing total flow time on a single
  machine.
\newblock SIAM Journal on Computing \textbf{28}(4), 1155--1166 (1999)

\bibitem{khouja2008}
Khouja, M., Goyal, S.: A review of the joint replenishment problem literature:
  1989--2005.
\newblock European Journal of Operational Research \textbf{186}(1), 1--16
  (2008)

\bibitem{lenstra1977}
Lenstra, J.K., Kan, A.R., Brucker, P.: Complexity of machine scheduling
  problems.
\newblock In: Annals of discrete mathematics, vol.~1, pp. 343--362. Elsevier
  (1977)

\bibitem{levi2008}
Levi, R., Roundy, R., Shmoys, D., Sviridenko, M.: A constant approximation
  algorithm for the one-warehouse multiretailer problem.
\newblock Management Science \textbf{54}(4), 763--776 (2008)

\bibitem{levi2006}
Levi, R., Roundy, R.O., Shmoys, D.B.: Primal-dual algorithms for deterministic
  inventory problems.
\newblock Mathematics of operations research \textbf{31}(2), 267--284 (2006)

\bibitem{levi2006improved}
Levi, R., Sviridenko, M.: Improved approximation algorithm for the
  one-warehouse multi-retailer problem.
\newblock In: Approximation, Randomization, and Combinatorial Optimization.
  Algorithms and Techniques, pp. 188--199. Springer (2006)

\bibitem{nonner2009}
Nonner, T., Souza, A.: A 5/3-approximation algorithm for joint replenishment
  with deadlines.
\newblock In: International Conference on Combinatorial Optimization and
  Applications, pp. 24--35. Springer (2009)

\bibitem{potts1980analysis}
Potts, C.N.: Analysis of a heuristic for one machine sequencing with release
  dates and delivery times.
\newblock Operations Research \textbf{28}(6), 1436--1441 (1980)

\bibitem{starr1962}
Starr, M.K., Miller, D.W.: Inventory control: theory and practice.
\newblock Prentice-Hall (1962)

\end{thebibliography}


\end{document}